\documentclass[12pt,regno]{amsart}
\usepackage{hyperref}
\usepackage{latexsym}
\usepackage{amssymb,amsfonts,amsmath}
\usepackage{graphicx}
\usepackage{indentfirst}
\usepackage{amsmath}
\usepackage{amsfonts}
\usepackage{amssymb}%
\usepackage{comment}
\usepackage{mathtools}
\usepackage{amsthm}
\usepackage{enumerate}
\usepackage{extarrows}
\usepackage{relsize}

\usepackage[cmtip,all]{xy}
\usepackage{verbatim}
\usepackage{tikz-cd}
\usepackage{mathtools}
\usepackage{mathrsfs}
\usepackage{geometry}
\DeclareMathAlphabet{\mathpzc}{OT1}{pzc}{m}{it}

\newcommand{\beq}{\begin{equation}}
\newcommand{\eeq}{\end{equation}}
\newcommand{\bear}{\begin{eqnarray}}
\newcommand{\eear}{\end{eqnarray}}

\newcommand*{\defeq}{\mathrel{\vcenter{\baselineskip0.5ex \lineskiplimit0pt
                     \hbox{\scriptsize.}\hbox{\scriptsize.}}}%
                     =}

\ifx\pdfoutput\undefined
\else
\fi
\hypersetup{colorlinks=false,bookmarksopen,bookmarksnumbered,citecolor=blue,
pdfstartview=FitH}

\parskip=\medskipamount
\def\a{\alpha}

\def\b{\beta}

\newcommand{\be}{\begin{equation}}
\newcommand{\ee}{\end{equation}}
\newcommand{\bea}{\begin{eqnarray}}
\newcommand{\eea}{\end{eqnarray}}

\newcommand{\ba}{\begin{array}}
\newcommand{\ea}{\end{array}}

\def\double #1{#1{\hbox{\kern-2pt $#1$}}}

\newcommand{\bsubeq}{\begin{subequations}}
\newcommand{\esubeq}{\end{subequations}}

\newcommand{\virgolette}{``}

\theoremstyle{definition}
\newtheorem{defn}{Definition}
\theoremstyle{plain}
\newtheorem{thm}{Theorem}
\newtheorem{prop}[thm]{Proposition}
\theoremstyle{remark}
\newtheorem*{remark}{Remark}

\begin{document}

\begin{flushright}
\par\end{flushright}
\vspace{1cm}

\title{Generalised Cocycles and Super $p$-Branes}
\date{}

\author{C.\ A.\ Cremonini}
\author{P.\ A.\ Grassi}

\email{carlo.alberto.cremonini@gmail.com, pietro.grassi@uniupo.it}

\address{Faculty of Mathematics and Physics, Mathematical Institute, Charles University Prague, Sokolovska 49/83, 186 75 Prague} 
\address{Dipartimento di Scienze e Innovazione Tecnologica (DiSIT), Universit\`a del Piemonte Orientale, viale T. Michel, 11, 15121 Alessandria, Italy}
\address{INFN, Sezione di Torino, via P. Giuria 1, 10125 Torino, Italy}
\address{Arnold Regge Center, Via P. Giuria 1, 10125, Torino, Italy}

\begin{abstract}

We study the cohomology (cocycles) of Lie superalgebras for the generalised complex of forms: superforms, pseudoforms and integral forms. We argue that these cocycles might be interpreted in the light of a new brane scan as generators of new higher-WZW terms and might provide new sources for supergravity. We use the technique of spectral sequences to abstractly compute the Chevalley-Eilenberg cohomology. We first focus on the superalgebra $\mathfrak{osp}(2|2)$ and show that there exist non-empty cohomology spaces among pseudoforms related to sub-superalgebras. We then extend some classical theorems by Koszul to include pseudoforms and integral forms. Further, we conjecture that the Poincar\'e duality extends to Lie superalgebras, as long as all the complexes of forms are taken into account and we prove that this holds for $\mathfrak{osp}(2|2)$. We finally construct the cohomology representatives explicitly by using a distributional realisation of pseudoforms and integral forms. On one hand, these results show that the cohomology of Lie superalgebras is larger than expected; on the other hand, we show the emergence of completely new cohomology classes represented by pseudoforms. These classes represent integral form classes of sub-superstructures.

\end{abstract}

\maketitle
\setcounter{footnote}{0}
\tableofcontents

\vfill
\eject
\section*{Introduction} \setcounter{equation}{0}

The advent of Lie superalgebras \cite{Kac:1977em} was impacting not only for new mathematical studies but also 
in the physical realm of supersymmetric theories. Since then, several applications and discoveries have used the natural 
generalization of the Lie algebra framework to graded Lie algebras, and in particular, several studies were carried out to classify all possible interesting examples 
of Lie superalgebras \cite{Grozman:1997xn,Kac:1977em,Leites:1982ia,Manin:1988ds}. As for Lie algebras, one of these classification techniques involves 
algebraic invariants and the corresponding Chevalley-Eilenberg cohomology \cite{CE,Koszul}. Therefore, it 
appeared straightforward to translate the Chevalley-Eilenberg analysis to the context of Lie superalgebras, as in the pioneering work \cite{Fuks}. 

The main difference between Lie superalgebras and Lie algebras is the presence of anti-commuting (odd) generators.
For a complete dictionary on superalgebras, we refer to \cite{Frappat}. 
Associated with those anti-commuting generators, one has the dual Maurer-Cartan forms which appear as commuting differential forms, in contrast 
to the usual anti-commuting 1-forms associated with the even generators. Using the Maurer-Cartan equations,  the Chevalley-Eilenberg cohomology is easily defined by following \cite{CE,Koszul} on the 
space of differential forms, namely on \emph{superforms}. General theorems for Lie algebras (see, e.g., \cite{GHV}) point out that there is a correspondence between the algebraic invariants and the cohomology classes. 
For Lie superalgebras, the crucial point 
is that the full cohomology does not coincide with the CE-cohomology of superforms only, but rather it embraces new 
types of forms which are known, in the super-geometric context, as {\it integral forms} and {\it pseudoforms} \cite{lei,CGN,Manin:1988ds,OP,Pen,Witten}. The 
former have been introduced on supermanifolds to implement a meaningful integration theory \cite{Manin:1988ds,lei} 
and their formal properties were studied in the seminal papers \cite{BS,Belopolsky:1996cy,Belopolsky:1997bg,Belopolsky:1997jz,FMS,VV}, using a distributional realisation. These new forms define a 
new complex on which the CE differential can be defined. This construction has been explored in \cite{CCGN2} where relevant details are discussed. 

The superform cohomology has been well known since the '80s (see, e.g., \cite{Fuks}) and, recently, has been explored in several works (see, e.g., \cite{boe,Lebedev:2004mq,Scheunert:1997ks,pro,Su:2019znq}) and the integral form cohomology was introduced in \cite{CCGN2}. The exploration of the cohomology in pseudoform complexes, in the way we will present in the paper, has never been carried out and it is the aim of the present work. A presentation of pseudoform cohomology (more in general of cohomology in any complex) was given in \cite{Su:2019znq}, where the authors introduced \virgolette mixed Fock spaces" from \virgolette mixing sets", computing the cohomology with respect to a BRST-like operator. This point of view of introducing generalised forms is inherited from \cite{Witten}. We will use as a guiding example the superalgebra $\mathfrak{osp}(2|2)$. We will give the definition of pseudoforms in terms of infinite-dimensional modules constructed from sub-(super)algebras, then we will define a notion of cohomology of such complexes.

The introduction of pseudoform modules is far from trivial since despite their recent use in physical contexts lead to unexpected results (see, e.g., \cite{CGN,CGN2,CCGN,CG,CG2}), their full mathematical understanding is still lacking. In this paper, we want to emphasise, at least in an algebraic context, how pseudoforms are strictly related to the existence of sub-(super)structures. In a geometric setting, this corresponds to (e.g., in \cite{Witten}) introducing sub-supermanifolds with non-trivial odd codimension. In \cite{Witten}, the author considers as a toy example $\mathbb{R}^{(0|*2)}$, locally parametrised by the two odd coordinates $\theta^1$ and $\theta^2$; he then uses the \virgolette odd constraint" $a \theta^1 + b \theta^2 = 0 , a,b \in \mathbb{C}$, to define a submanifold with odd codimension 1, and shows that the correct object to be integrated on such submanifold is actually a pseudoform (see also \cite{Wit2}). In the algebraic setting, the sub-(super)structures are already built-in: they are sub-superalgebras. We will leverage on these structures to show how pseudoforms emerge naturally. In particular, we will show that sub-(super)algebras play the analogous role to the constraints and that pseudoforms are related to the integral forms defined on sub-(super)algebras and (super)cosets. We will show that these constructions are independent of the formal distributional realisation, nonetheless, we will emphasise that it can be used, once again, as a powerful tool for inspiring calculations and suggesting results.

In the first section, we will collect some motivations that lead us to the systematic study of Lie superalgebra cohomology in all form sectors. In the second section, we will introduce some basic material to make the paper self-contained: we will introduce the Chevalley-Eilenberg cochain complex, relative cohomology, Koszul-Hochschild-Serre spectral sequence, integral forms in supergeometry and on Lie (super)algebras. We will then introduce pseudoform modules as infinite-dimensional modules constructed from sub-(super)algebras and then CE cohomology on pseudoforms. We will then comment on the meaning of the \virgolette picture number" that appears in supergeometry (see, e.g., \cite{Witten}) in terms of such modules. In the third section, we will consider the specific example $\mathfrak{osp}(2|2)$ to construct pseudoforms. We will use Koszul spectral sequence \cite{Koszul} (i.e., Hochschild-Serre spectral sequence \cite{HocSer} with trivial module) and, in particular, we will infer pseudoforms as integral forms of the (super)coset $\mathfrak{osp} \left( 2|2 \right) / \mathfrak{osp} \left( 1|2 \right)$ and of the sub-(super)algebra $\mathfrak{osp} \left( 1|2 \right)$. In the fourth section, we will introduce some general constructions, starting from the definition of a generalised spectral sequence for Lie superalgebras based on two filtrations. We will show that these two filtrations induce two inequivalent sectors of the first page of the spectral sequence when considering Lie sub-(super)algebras, but reduce to a single sector when considering purely even Lie sub-algebras. We will conclude the section with the extension of classical theorems for Lie algebras (see \cite{Koszul,HocSer}) to the super setting, for any form complex. In the fifth section, we will show how pseudoforms of the main example $\mathfrak{osp}(2|2)$ emerge by using the explicit distributional realisation introduced in the first section, when considering the purely odd coset $\mathfrak{osp}(2|2) / \left( \mathfrak{sp} (2) \times \mathfrak{so}(2) \right)$ and we will then construct invariant pseudoforms explicitly. We will support explicit calculations with Poincar\'e polynomials, whose use in this algebraic setting is briefly reviewed.

\section{Cocycles and Brane Scan}

The reason to study CE cohomology for Lie superalgebras is rooted in the application 
to supergravity, superstring and super Dp-branes. It has been worldwide recognized that several supergravity 
backgrounds are constructed in terms of Lie superalgebras and their cosets. As an example, let us consider the case of flat supergravity background, the underlying manifold is a supermanifold (either as 
the worldvolume for N=1,2,4 Ramond-Neveu-Schwarz superstrings or Dp-branes \cite{GSW}, or as 
a target space, as in type IIA/B for d=10 supergravity) $\mathcal{SM}^{(m|n)}$ where $m$ is the bosonic (even) dimension and $n$ is the fermionic (odd) dimension. Locally, the supermanifold is described by an atlas with charts 
$({\mathbb R}^{(m|n)}, x^i, \theta^\alpha)$ where $x^i$ and $\theta^\alpha$ are the even and odd coordinates, respectively. As usual, we can build the cotangent bundle starting from the (parity changed) 1-forms $(V^i, \psi^\alpha)$ (which are conveniently expressed in terms of fundamental 1-forms $(dx^i, d\theta^\alpha)$). For superstrings and Dp-branes 
the latter represent dynamical fields with values in the cotangent bundle of the supermanifold, in supergravity $ (V^i, \psi^\alpha)$ are the dynamical super vielbein. To construct the action or the observables for superstrings or Dp-branes one needs to uncover all possible cocycles $\omega^p(V,\psi)$ in order to establish the existence of $\kappa$-invariant action functionals. For that one needs the pullback of those cocycles on the worldvolume respecting the isometries of the manifold (see for example the textbook \cite{GSW} or \cite{ACHU,POL,DUFF1}). The analysis of these cocycles corresponds to the computation of the CE cohomology groups of super translations, representing graded non-semi simple Lie algebras. In the case of supergravity, the fundamental ingredient is the action functional whose Euler-Lagrangian equations provide the complete set of constraints and equations of motion. That action functional is constructed in terms of the cocycles associated 
with the super isometries of the underlying supermanifold. This framework, inspired by the early work \cite{SUL}, is also known as Free Differential Algebra leading to 
the complete spectrum of a supergravity model \cite{CDF}. Therefore, cocycles are pivotal in both frameworks, but 
the real challenge is represented by models on non-trivial backgrounds (such as superstrings on $AdS_5 \times S^5$). In several 
cases, the supermanifold is the group-manifold associated with a super Lie algebra, or a coset space (as 
$PSU(2,2|4)/SO(4,1) \times SO(5)$ for  $AdS_5 \times S^5$, $OSp(8|4)/SO(7)\times Sp(4)$ for $AdS_4 \times S^7$, 
or $OSp(6|4)/ U(3)\times Sp(4)$ for $AdS_4 \times \mathbb{CP}^3$) and the associated cocycles 
have to be computed in terms of the CE cohomology of the corresponding Lie algebras or cosets. Some of the cocycles for several models have been already computed (see, e.g., \cite{DUFF2}), but a systematic analysis is missing (except for super translations, wherein \cite{FSS} the authors use FDAs and their dual version of $L_\infty$-algebras to classify higher cocycles).

In the present work, we find that there are additional cocycles associated with pseudoforms, which are forms living in a different complex w.r.t. usual superforms, and therefore it would be interesting, in the future, to study the associated Dp-brane interpretation or supergravity solutions. In particular, given a cohomology class $\omega^{(p|q)} \in H^{(p|q)} \left( \mathfrak{g} \right)$, by following the FDA algorithm we can introduce a \virgolette potential" $\eta^{(p-1|q)}$ s.t.
\begin{equation}\label{AQGSBA}
	d \eta^{(p-1|q)} = \omega^{(p|q)} \ .
\end{equation}
It will be our interest to study these new generators in the contexts described above, to see if pseudoforms can give rise to new, unexplored, classes of branes and supergravity couplings.

\section{Preliminaries}\label{Preliminaries}

In this section, we recall some basic definitions of Lie (super)algebras, Chevalley-Eilenberg complexes, Lie (super)algebra cohomology and spectral sequences, to make the paper self-contained. We refer the reader to \cite{CE,Koszul,HocSer} for more detailed expositions in the standard Lie-algebraic setting and to \cite{Fuks} for the Lie-superalgebraic setting. Then, we introduce the notions of \virgolette pseudoforms" and \virgolette picture number". In particular, pseudoforms will be the main point of interest in the present paper.

\begin{defn}[Chevalley-Eilenberg chain complex]
	Given a (finite dimensional) Lie (super)algebra $\mathfrak{g}$ over the field $\mathbb{K}$ of characteristic zero \footnote{Throughout the paper we will systematically assume $\mathbb{K} = \mathbb{R}, \mathbb{C}$.} and a $\mathfrak{g}$-module $V$, we define \emph{p-chains} of $\mathfrak{g}$ valued in $V$ as (graded)alternating $\mathbb{K}$-linear products
	\begin{equation}\label{PA}
		C_p \left( \mathfrak{g} , V \right) \defeq \wedge^p \mathfrak{g} \otimes V = \bigoplus_{r=1}^p  \left( \wedge^r \mathfrak{g}_0 \otimes S^{p-r} \mathfrak{g}_1 \right) \otimes V \ .
	\end{equation}
	Throughout the paper we will denote with a $\mathfrak{g}$ basis $\left\lbrace \mathcal{Y}_a \right\rbrace = \left\lbrace X_i , \xi_\alpha \right\rbrace$, where $i=1,\ldots, \dim \mathfrak{g}_0$ and $\alpha = 1,\ldots,\dim \mathfrak{g}_1$. We can lift \eqref{PA} to a \emph{chain complex} by introducing the differential $\partial: C_p \left( \mathfrak{g} , V \right) \to C_{p-1} \left( \mathfrak{g} , V \right)$ acting on $f \otimes \left( \mathcal{Y}_{a_1} \wedge \ldots \wedge \mathcal{Y}_{a_p} \right) \in C_p \left( \mathfrak{g} , V \right)$ as
	\begin{eqnarray}
		\nonumber \partial \left[ f \otimes \left( \mathcal{Y}_{a_1} \wedge \ldots \wedge \mathcal{Y}_{a_p} \right) \right] &=& \sum_{i=1}^p (-1)^{\delta_i} \left( \mathcal{Y}_{a_i} f \right) \otimes \left( \mathcal{Y}_{a_1} \wedge \ldots \wedge \hat{\mathcal{Y}}_{a_i} \wedge \ldots \wedge \mathcal{Y}_{a_p} \right) + \\
		\label{PB} &+& \sum_{i<j}^p (-1)^{\delta_{i,j}} f \otimes \left( \left[ \mathcal{Y}_{a_i} , \mathcal{Y}_{a_j} \right] \wedge \mathcal{Y}_{a_1} \wedge \ldots \wedge \hat{\mathcal{Y}}_{a_i} \wedge \ldots \wedge \hat{\mathcal{Y}}_{a_j} \wedge \ldots \wedge \mathcal{Y}_{a_p} \right) \ ,
	\end{eqnarray}
	where $\delta_i = \left| \pi \mathcal{Y}_{a_i} \right| \left( \sum_{p=1}^{i-1} \left| \pi \mathcal{Y}_{a_p} \right| + \left| f \right| \right)$, $\delta_{i,j} = \left| \pi \mathcal{Y}_{a_i} \right| \sum_{p=1}^{i-1} \left| \pi \mathcal{Y}_{a_p} \right| + \left| \pi \mathcal{Y}_{a_j} \right| \sum_{p=1}^{j-1} \left| \pi \mathcal{Y}_{a_p} \right| - \left| \pi \mathcal{Y}_{a_i} \right| \left| \pi \mathcal{Y}_{a_j} \right|$ and where with the symbol $\hat{\mathcal{Y}}$ we indicate that the vector is omitted; the nilpotence of the operator $\partial$ is a consequence of the (graded) Jacobi identities. The \emph{Chevalley-Eilenberg chain complex} is then denoted as $\left( C_p \left( \mathfrak{g} , V \right) , \partial \right)$.
\end{defn}

In the following sections we will be interested in the dual concept, namely the Chevalley-Eilenberg cochain complex:
\begin{defn}[Chevalley-Eilenberg cochain complex]
	We define Chevalley-Eilenberg \emph{p-cochains} of $\mathfrak{g}$ valued in $V$ as (graded)alternating $\mathbb{K}$-linear maps
	\begin{equation}\label{PC}
		C^p \left( \mathfrak{g} , V \right) \defeq Hom_{\mathbb{K}} \left( \wedge^p \mathfrak{g} , V \right) = \wedge^p \mathfrak{g}^* \otimes V = \bigoplus_{r=1}^p \left( \wedge^r \mathfrak{g}_0^* \otimes S^{q-r} \mathfrak{g}_1^* \right) \otimes V \ .
	\end{equation}
	We can lift \eqref{PC} to a \emph{cochain complex} by introducing the differential $d: C^p \left( \mathfrak{g} , V \right) \to C^{p+1} \left( \mathfrak{g} , V \right)$ defined on $\omega \in C^p \left( \mathfrak{g} , V \right)$ by
	\begin{eqnarray}
		\nonumber d \omega \left( \mathcal{Y}_{a_1} , \ldots , \mathcal{Y}_{a_{p+1}} \right) &=& \sum_{i=1}^{p+1} (-1)^{\delta_i} \mathcal{Y}_{a_i} \omega \left( \mathcal{Y}_{a_1} \wedge \ldots \wedge \hat{\mathcal{Y}}_{a_i} \wedge \ldots \wedge \mathcal{Y}_{a_{p+1}} \right) + \\
		\label{PD} &+& \sum_{i<j} (-1)^{\delta_{i,j}} \omega \left( \left[ \mathcal{Y}_{a_i} , \mathcal{Y}_{a_j} \right] \wedge \mathcal{Y}_{a_1} \wedge \ldots \wedge \hat{\mathcal{Y}}_{a_i} \wedge \ldots \wedge \hat{\mathcal{Y}}_{a_j} \wedge \ldots \wedge \mathcal{Y}_{a_{p+1}} \right) \ ,
	\end{eqnarray}
	where $\delta_i$ and $\delta_{i,j}$ are defined as above; again, the nilpotence of the operator $d$ is a consequence of the (graded) Jacobi identities. The \emph{Chevalley-Eilenberg cochain complex} is then denoted as $\left( C^p \left( \mathfrak{g} , V \right) , d \right)$.
\end{defn}

In the following sections, we will deal mainly with the case of the trivial module $V = \mathbb{K}$ and we will adopt the notation
\begin{equation}\label{PDA}
	C^p \left( \mathfrak{g} , \mathbb{K} \right) \equiv C^p \left( \mathfrak{g} \right) \equiv \Omega^p \left( \mathfrak{g} , \mathbb{K} \right) \equiv \Omega^p \left( \mathfrak{g} \right) \ .
\end{equation}

\begin{defn}[Lie (super)algebra cohomology]
	Given a Chevalley-Eilenberg cochain complex of a Lie (super)algebra $\mathfrak{g}$ and a $\mathfrak{g}$-module $V$, we define the space of \emph{p-cocycles} (or \emph{closed forms})
	\begin{equation}\label{PE}
		Z^p \left( \mathfrak{g} , V \right) \defeq \left\lbrace \omega \in C^p \left( \mathfrak{g} , V \right) : d \omega = 0 \right\rbrace \ .
	\end{equation}
	The space of \emph{p-coboundaries} (or \emph{exact forms}) is defined as
	\begin{equation}\label{PF}
		B^p \left( \mathfrak{g} , V \right) \defeq \left\lbrace \omega \in C^p \left( \mathfrak{g} , V \right) : \exists \rho \in C^{p-1} \left( \mathfrak{g} , V \right) : d \rho = \omega \right\rbrace \ .
	\end{equation}
	Because of the nilpotence of the operator $d$, we can consistently define the \emph{p-cohomology group as}
	\begin{equation}\label{PG}
		H^p \left( \mathfrak{g} , V \right) \defeq \frac{Z^p \left( \mathfrak{g} , V \right)}{B^p \left( \mathfrak{g} , V \right)} \ .
	\end{equation}
\end{defn}

\begin{defn}[Relative Lie (super)algebra cohomology]\label{Relative Lie (super)algebra cohomology}
	Given a Lie (super)algebra $\mathfrak{g}$ and a Lie sub-(super)algebra $\mathfrak{h}$, we define the space of \emph{horizontal p-cochains} with values in a module $V$ as
	\begin{equation}\label{PH}
		C^p_{hor} \left( \mathfrak{g} , \mathfrak{h} , V \right) \equiv C^p \left( \mathfrak{k} , V \right) \defeq \left\lbrace \omega \in C^p \left( \mathfrak{g} , V \right) : \omega \left( \mathcal{Y}_a , \mathcal{Y}_{a_1} , \ldots , \mathcal{Y}_{a_{p-1}} \right) = 0 , \forall \mathcal{Y}_a \in \mathfrak{h}, \mathcal{Y}_{a_1} , \ldots , \mathcal{Y}_{a_{p-1}} \in \mathfrak{g} \right\rbrace \ .
	\end{equation}
	We define the space of \emph{$\mathfrak{h}$-invariant p-cochains} with values in $V$ as
	\begin{equation}\label{PI}
		\left( C^p \left( \mathfrak{g} , V \right) \right)^{\mathfrak{h}} \defeq \left\lbrace \omega \in C^p \left( \mathfrak{g} , V \right) : \sum_{j=1}^p (-1)^{\delta_{i,j}} \omega \left( \mathcal{Y}_{a_1} , \ldots , \left[ \mathcal{Y}_a , \mathcal{Y}_{a_j} \right] , \ldots , \mathcal{Y}_{a_p} \right)  = 0 , \forall \mathcal{Y}_a \in \mathfrak{h}, \mathcal{Y}_{a_1} , \ldots , \mathcal{Y}_{a_p} \in \mathfrak{g} \right\rbrace \ ,
	\end{equation}
	where $\delta_{i,j} = \left| \pi \mathcal{Y}_{a} \right| \sum_{r=1}^{j-1} \left| \pi \mathcal{Y}_{a_r} \right|$. The forms which are both horizontal and $\mathfrak{h}$-invariant are called \emph{basic} and we will denote them as (we denote $\mathfrak{k}=\mathfrak{g}/\mathfrak{h}$ as a quotient of vector spaces, since $\mathfrak{h}$ in general is not an ideal, or analogously $\mathfrak{g} = \mathfrak{h} \oplus \mathfrak{k}$)
	\begin{equation}\label{PJ}
		\left( C^p \left( \mathfrak{k} , V \right) \right)^{\mathfrak{h}} \equiv \left( C^p \left( \mathfrak{g}/\mathfrak{h} , V \right) \right)^{\mathfrak{h}} \ .
	\end{equation}
	Cocycles and coboundaries are defined by taking the restriction of \eqref{PE} and \eqref{PF} to \eqref{PJ}; analogously, we can restrict the differential \eqref{PD} to basic forms, denoting it $\left. d \right|_{basic} \equiv \nabla_{\mathfrak{k}}$, define closed and exact forms and finally the \emph{cohomology of $\mathfrak{g}$ relative to $\mathfrak{h}$}:
	\begin{equation}\label{PK}
		H^\bullet \left( \mathfrak{g}, \mathfrak{h}, V \right) \defeq \frac{\left\lbrace \omega \in \left( C^\bullet \left( \mathfrak{k} , V \right) \right)^{\mathfrak{h}} : \nabla_{\mathfrak{k}} \omega = 0 \right\rbrace}{\left\lbrace \omega \in \left( C^\bullet \left( \mathfrak{k} , V \right) \right)^{\mathfrak{h}} : \exists  \eta \in \left( C^\bullet \left( \mathfrak{k} , V \right) \right)^{\mathfrak{h}} , \omega = \nabla_{\mathfrak{k}} \eta \right\rbrace} \ .
	\end{equation}
	Again, when $V = \mathbb{K}$, we will denote the cohomology cochain complex as
	\begin{equation}\label{PL}
		H^\bullet \left( \mathfrak{g}, \mathfrak{h}, \mathbb{K} ; \nabla_{\mathfrak{k}} \right) \equiv H^\bullet \left( \mathfrak{g}, \mathfrak{h}, \mathbb{K} \right) \equiv H^\bullet \left( \mathfrak{g}, \mathfrak{h} \right) \ .
	\end{equation}
\end{defn}
We can adopt a more modern notation for horizontal and $\mathfrak{h}$-invariant forms: \eqref{PH} and \eqref{PI} can be rewritten as
\begin{equation*}
	C^p \left( \mathfrak{k} , V \right) \defeq \left\lbrace \omega \in C^p \left( \mathfrak{g} , V \right) : \iota_{\mathcal{Y}_a} \omega = 0 , \forall \mathcal{Y}_a \in \mathfrak{h} \right\rbrace \ ,
\end{equation*}
\begin{equation}
	\left( C^p \left( \mathfrak{g} , V \right) \right)^{\mathfrak{h}} \defeq \left\lbrace \omega \in C^p \left( \mathfrak{g} , V \right) : \mathcal{L}_{\mathcal{Y}_a} \omega = 0 , \forall \mathcal{Y}_a \in \mathfrak{h} \right\rbrace \ .
\end{equation}
respectively, where $\iota_{\mathcal{Y}_a}$ denotes the contraction along the vector $\mathcal{Y}_a$ and $\mathcal{L}_{\mathcal{Y}_a}$ denotes the Lie derivative along the vector $\mathcal{Y}_a$.

In the case $V = \mathbb{K}$, we can use a more convenient form for the differential \eqref{PD}; in particular, the first term drops, and, given a $\mathfrak{g}$ basis $\left\lbrace \mathcal{Y}_a \right\rbrace = \left\lbrace X_i , \xi_\alpha \right\rbrace$ and the (parity changed) dual $\Pi \mathfrak{g}^*$ basis $\left\lbrace \mathcal{Y}^{*a} \right\rbrace = \left\lbrace V^i , \psi^\alpha \right\rbrace$, $i=1 , \ldots , \dim \mathfrak{g}_0$, $\alpha = 1 , \ldots , \dim \mathfrak{g}_1$, the differential reads
\begin{equation}
	d = \sum_{a , b , c} f^a_{b c} \mathcal{Y}^{*b} \wedge \mathcal{Y}^{*c} \iota_{\mathcal{Y}_a} \ ,
\end{equation}
where $f^a_{b c}$ are the structure constants of the Lie (super)algebra $\mathfrak{g}$.

\subsection{Koszul-Hochschild-Serre Spectral Sequence}

Here we will briefly review the main method of constructing the cohomology of Lie (super)algebras we will use in the forthcoming sections. For more details, we refer to the seminal papers \cite{Koszul} (where the author considers the cohomology of Lie algebras with values in a trivial module) and \cite{HocSer} (where the authors generalise to general modules) and to the vast book \cite{Fuks} (where the author extend many results to the super-algebraic setting).

\begin{defn}[Koszul-Hochschild-Serre Spectral Sequence]\label{Koszul-Hochschild-Serre Spectral Sequence}
	Given a Lie (super) algebra $\mathfrak{g}$, a sub-(super)algebra $\mathfrak{h}$ and a $\mathfrak{g}$-module $V$, we define the filtration
	\begin{equation}\label{KHSSSA}
		F^p C^q \left( \mathfrak{g} , V \right) = \left\lbrace \omega \in C^q \left( \mathfrak{g} , V \right) : \forall \mathcal{Y}_a \in \mathfrak{h} , \iota_{\mathcal{Y}_{a_1}} \ldots \iota_{\mathcal{Y}_{a_{q+1-p}}} \omega = 0 \right\rbrace \ , p,q \in \mathbb{Z} \ ,
	\end{equation}
	where we denote $C^{q<0} \left( \mathfrak{g} , V \right) = 0$. In particular, for $F^p C^q \left( \mathfrak{g} , V \right)$ we have
	\begin{equation}\label{KHSSSB}
		F^{q+n+1} C^q \left( \mathfrak{g} , V \right) = 0 \ , \ F^{p+1} C^q \left( \mathfrak{g} , V \right) \subseteq F^p C^q \left( \mathfrak{g} , V \right) \ , $$ $$ d F^p C^q \left( \mathfrak{g} , V \right) \subseteq F^p C^{q+1} \left( \mathfrak{g} , V \right) \ , \forall n \in \mathbb{N} \cup \left\lbrace 0 \right\rbrace \ , \ \forall p,q \in \mathbb{Z} \ .
	\end{equation}
	Associated to the filtration \eqref{KHSSSA}, there exist a spectral sequence $\left( E_s^{\bullet , \bullet} , d_s \right)_{s \in \mathbb{N}\cup \left\lbrace 0 \right\rbrace}$ that converges to $H^\bullet \left( \mathfrak{g} , V \right)$. More explicitly, this means that we define \emph{page zero} of the spectral sequence as
	\begin{equation}\label{KHSSSC}
		E_0^{p,q} \defeq F^p C^{p+q} \left( \mathfrak{g} , V \right) / F^{p+1} C^{p+q} \left( \mathfrak{g} , V \right) \ .
	\end{equation}
	The (nilpotent) differentials $d_s$ are induced by the Chevalley-Eilenberg differential \eqref{PD} and each page of the spectral sequence is defined as the cohomology of the previous one:
	\begin{equation}\label{KHSSSD}
		d_s : E_s^{p,q} \to E_s^{p+s,q+1-s} \ , \ E_{s+1}^{\bullet , \bullet} \defeq \left( E_{s}^{\bullet , \bullet} , d_s \right) \ , \ E_\infty^{\bullet,\bullet} \cong H^\bullet \left( \mathfrak{g} , V \right) \ .
	\end{equation}
	If $\exists n \in \mathbb{N}\cup \left\lbrace 0 \right\rbrace : d_s = 0 , \forall s \geq n$, then we say that the spectral sequence \emph{converges at page} $n$ and we denote
	\begin{equation}\label{KHSSSE}
		E_{n}^{\bullet , \bullet} = E_{n+1}^{\bullet , \bullet} = \ldots = E_{\infty}^{\bullet , \bullet} \ .
	\end{equation}
\end{defn}

\begin{thm}[Hochschild-Serre]\label{KHStheorem}
	Given a Lie algebra $\mathfrak{g}$ over a (characteristic zero) field $\mathbb{K}$, a sub-algebra $\mathfrak{h}$ reductive in the ambient Lie algebra and a finite-dimensional $\mathfrak{g}$-module $V$, the first pages of the spectral sequence read
	\begin{eqnarray}
		\label{KHSSSF} E_1^{p,q} &\cong& H^q \left( \mathfrak{h} , \mathbb{K} \right) \otimes C^p \left( \mathfrak{g}/\mathfrak{h} , V \right)^{\mathfrak{h}} \ , \\
		\label{KHSSSG} E_2^{p,q} &\cong& H^q \left( \mathfrak{h} , \mathbb{K} \right) \otimes H^p \left( \mathfrak{g} , \mathfrak{h} , V \right) \ .
	\end{eqnarray}
\end{thm}

\subsection{Integral Forms on Lie Superalgebras}

Here, we briefly recall some definitions and results regarding \emph{integral forms} on supermanifolds and specialise them to Lie superalgebras. For a clear account of integral forms see, e.g., \cite{Manin:1988ds}, for their introduction in the (super)algebraic context see \cite{CCGN2}.

\begin{defn}[Integral Forms]
	Given a supermanifold $\mathcal{SM} = \left( \left| \mathcal{SM} \right| , \mathcal{O} \right)$, i.e., a topological space $\left| \mathcal{SM} \right|$ equipped with a sheaf of supercommutative rings $\mathcal{O}$, we define its \emph{Berezinian sheaf} as
	\begin{equation}\label{IFLSA}
		\mathpzc{B}er \left( \mathcal{SM} \right) \defeq \left( \mathpzc{B}er \Omega^1_{odd} \left( \mathcal{SM} \right) \right)^* \ .
	\end{equation}
	This sheaf is locally generated by taking the product of all the odd generators of the tangent and cotangent spaces of $\mathcal{SM}$ (we denote by $z^a = \left\lbrace x^i , \theta^\alpha \right\rbrace$ a local system of coordinates):
	\begin{equation}\label{IFLSB}
		\mathpzc{B}er \left( \mathcal{SM} \right) \cong \mathcal{O} \cdot \left[ \bigwedge_{i=1}^{\dim \mathcal{SM}_0} dx^i \otimes \bigwedge_{\alpha=1}^{\dim \mathcal{SM}_1} \partial_{\theta^\alpha} \right] \ .
	\end{equation}
	We can define the \emph{integral forms}, which will be denoted as $C^\bullet_{int} \left( \mathcal{SM} \right) \equiv \Omega^{\bullet|\dim \mathcal{SM}_1} \left( \mathcal{SM} \right)$, as (graded) symmetric powers of parity changed vector fields with values in the Berezinian bundle:
	\begin{equation}\label{IFLSC}
		C^\bullet_{int} \left( \mathcal{SM} \right) \defeq \mathpzc{B}er \left( \mathcal{SM} \right) \otimes_{\mathcal{O}} S^{\dim \mathcal{SM}_0 - \bullet} \Pi T \mathcal{SM}  \ .
	\end{equation}
	We can lift integral forms to a cochain complex by introducing the differential
	\begin{equation}\label{IFLSD}
		\delta : C^p_{int} \left( \mathcal{SM} \right) \to C^{p+1}_{int} \left( \mathcal{SM} \right) $$ $$
		\delta = \sum_a \mathcal{L}^R_{\frac{\partial}{\partial z^a}} \otimes_{\mathbb{K}} \frac{\partial}{\partial \left( \pi \frac{\partial}{\partial z^a} \right)} (-1)^{\left| z^a \right|} \ ,
	\end{equation}
	where $\mathcal{L}^R$ denotes the (right) Lie derivative and it is possible to show that the definition of $\delta$ is independent on the coordinate choice.
\end{defn}

By following \cite{Belopolsky:1996cy,Belopolsky:1997bg,Belopolsky:1997jz} and \cite{Witten}, we can give a different realisation of integral forms, inspired by string theory (see \cite{FMS}), in terms of \emph{Dirac distributions} (see also \cite{CG,CG2} for the use of this realisation in Chern-Simons theory and BV): we locally realise the Berezinian sheaf as
\begin{equation}\label{IFLSE}
	\mathpzc{B}er \left( \mathcal{SM} \right) \equiv \omega^{(\dim \mathcal{SM}_0 | \dim \mathcal{SM}_1)} \cong \mathcal{O} \cdot \left[ \bigwedge_{i=1}^{\dim \mathcal{SM}_0} dx^i \wedge \bigwedge_{\alpha=1}^{\dim \mathcal{SM}_1} \delta \left( d \theta^\alpha \right) \right] \ ,
\end{equation}
where $\delta \left( d \theta \right)$ are (formal) Dirac delta distributions. The symbol  $\delta \left( d \theta \right)$ satisfies the following distributional identities
\begin{eqnarray}\label{IFLSF}
	 d \theta \delta \left( d \theta \right) = 0 \ , \ \delta \left( \lambda d \theta \right) = \frac{1}{\lambda} \delta \left( d \theta \right) \ , \ d \theta \iota^{(p)} \delta \left( d \theta \right) = - p \iota^{p-1} \delta \left( d \theta \right) \equiv - p \delta^{(p-1)} \left( d \theta \right) \ , \\
	 \nonumber \delta \left( d \theta^\alpha \right) \wedge \delta \left( d \theta^\beta \right) = - \delta \left( d \theta^\beta \right) \wedge \delta \left( d \theta^\alpha \right) \ , \ dx \wedge \delta \left( d \theta \right) = - \delta \left( d \theta \right) \wedge d x \ , \ d \delta \left( d \theta \right) = d \left( d \theta \right) \wedge \iota \delta \left( d \theta \right) \ ,
	 \end{eqnarray}
indicating that actually these are not conventional distributions, but rather \emph{de Rham currents} (see \cite{Witten}). Integral forms are constructed exactly as in \eqref{IFLSC}, but now the parity changed vector fields are represented as contractions acting on \eqref{IFLSE}. Then, a generic integral form can be locally expressed as
\begin{equation}\label{IFLSG}
	\omega^{(p|\dim \mathcal{SM}_1)} = \omega^{[i_1 \ldots i_r](\alpha_1 \ldots \alpha_s)} \left( x , \theta \right) \iota_{i_1} \ldots \iota_{i_{r}} \iota_{\alpha_1} \ldots \iota_{\alpha_s} \left[ \bigwedge_{i=1}^{\dim \mathcal{SM}_0} dx^i \wedge \bigwedge_{\alpha=1}^{\dim \mathcal{SM}_1} \delta \left( d \theta^\alpha \right) \right] \ ,
\end{equation}
where $p=\dim \mathcal{SM}_0 - r - s$ counts the form number and $\dim \mathcal{SM}_1$ counts the \emph{picture number}, i.e., the number of $\delta$'s appearing in $\omega^{(p|\dim \mathcal{SM}_1)}$. This realisation of integral forms suggests the (local) construction of forms with non-maximal and non-zero number of delta's: the \emph{pseudoforms}. A general pseudoform with $q$ Dirac delta's (i.e., picture number $q$) is locally given by
\begin{equation}\label{IFLSH}
	\omega^{(p|q)} = \omega_{[a_1 \ldots a_r](\alpha_1 \ldots \alpha_s)[\beta_1 \ldots \beta_q]} \left( x , \theta \right) dx^{a_1} \wedge \ldots \wedge dx^{a_r} \wedge d \theta^{\alpha_1} \wedge \ldots \wedge d \theta^{\alpha_s} \wedge \delta^{(t_1)} \left( d \theta^{\beta_1} \right) \wedge \ldots \wedge \delta^{(t_q)} \left( d \theta^{\beta_q} \right) \ ,
\end{equation}
where $\delta^{(i)} \left( d \theta \right) \equiv \left( \iota \right)^i \delta \left( d \theta \right)$. The form number of \eqref{IFLSH} is obtained as 
\begin{equation}\label{IFLSI}
	p = r + s - \sum_{i=1}^q t_i \ ,
\end{equation}
since the contractions carry negative form number. The two numbers $p$ and $q$ in eq. \eqref{IFLSH}, corresponding to the form number and the picture number, respectively, range as $-\infty < p < +\infty$ and $0 \leq q \leq n$. If $q=0$, we have superforms, if $q=n$ we have integral forms, if $0<q<n$ we have pseudoforms. One of the advantages of the Dirac delta forms realisation consists in the use of a single differential for superforms, integral forms and pseudoforms: we can use the de Rham differential $d$ on every complex for any picture number, by supporting it with the formal properties in \eqref{IFLSF}.

Integral forms are essential on supermanifolds, to implement a consistent theory of integration; in the algebraic context of this paper, there is no notion of integration, but we will show that we can use them to introduce new algebraic invariants that are not kept into account when considering standard differential forms. Pseudoforms are only locally defined on supermanifolds (since they do not transform as tensors under superdiffeomorphisms) but they are the main topic of the present paper: since we are dealing with algebras, we are not concerned with changes of coordinates systems, hence their non-tensorial nature will not be an issue. Moreover, we will show how to introduce pseudoforms as integral forms of sub-structures (e.g., sub-(super)algebras), hence we will overcome their definition in terms of formal properties of Dirac delta distributions. We now give the definition of \emph{Haar Berezinian} (as it can be related to the measure of a supergroup related to the superalgebra $\mathfrak{g}$) and integral forms for the algebraic context.
\begin{defn}[Integral Forms]
	Given a Lie superalgebra $\mathfrak{g}$, with dimension $\dim \mathfrak{g} = (m|n)$, a basis $\mathcal{Y}_a = \left\lbrace X_i , \xi_\alpha \right\rbrace$, its (parity changed) dual $\Pi \mathfrak{g}^*$, with a basis $\mathcal{Y}^{*a} = \left\lbrace V^i , \psi^\alpha \right\rbrace$, $a=1,\ldots,m$ , $\alpha=1, \ldots, n$, and a $\mathfrak{g}$-module $V$, we define the \emph{Haar Berezinian} of $\mathfrak{g}$ with values  in $V$ as
	\begin{equation}\label{IFLSJ}
		\mathpzc{B}er \left( \mathfrak{g} \right) \defeq V \cdot \left[ \bigwedge_{i=1}^{m} V^i \otimes \bigwedge_{\alpha=1}^{n} \xi_\alpha \right] \equiv V \cdot \mathcal{D} \ .
	\end{equation}
	We define the integral form cochain complex as the sequence of spaces
	\begin{equation}\label{IFLSK}
		C^p_{int} \left( \mathfrak{g} , V \right) \defeq \mathpzc{B}er \left( \mathfrak{g} \right) \otimes S^{m-p} \Pi \mathfrak{g} = \mathcal{D} \otimes \left[ V \otimes S^{m-p} \Pi \mathfrak{g} \right] \ ,
	\end{equation}
	equipped with the differential defined in terms of \eqref{PB}
	\begin{eqnarray}
		\label{IFLSL} \delta: C^p_{int} \left( \mathfrak{g} , V \right) &\to& C^{p+1}_{int} \left( \mathfrak{g} , V \right) \\
		\nonumber \mathcal{D} \otimes \left[ f \otimes \mathcal{Y}_1 \wedge \ldots \wedge \mathcal{Y}_{m-p} \right] &\mapsto& \delta \left( \mathcal{D} \otimes \left[ f \otimes \mathcal{Y}_1 \wedge \ldots \wedge \mathcal{Y}_{m-p} \right] \right) = \mathcal{D} \otimes \partial \left[ f \otimes \mathcal{Y}_1 \wedge \ldots \wedge \mathcal{Y}_{m-p} \right] \ ,
	\end{eqnarray}
\end{defn}
which is clearly nilpotent. Then one can define the Chevalley-Eilenberg cohomology on integral forms in the obvious way. As already stated, in the present paper we will deal with the trivial module case only: $V= \mathbb{K}$. Moreover, we will be interested in the \virgolette basic classical Lie superalgebras", i.e., those classical Lie superalgebras admitting a non-degenerate bilinear form (see \cite{Frappat}). From these definitions we have that the integral Chevalley-Eilenberg cohomology $H^p_{int} (\mathfrak{g} , \mathbb{K})$ becomes a \emph{twist} of the $(n-p)$-th homology group by $\mathcal{D}$. Then the following proposition:
\begin{prop}[Berezinian Complement Quasi-Isomorphism]\label{Berezinian Complement Quasi-Isomorphism}
	Given a basic classical Lie (super)algebra $\mathfrak{g}$ of dimension $\dim \mathfrak{g} = (m|n)$, the map
	\begin{eqnarray}
		\label{IFLSM} \star : C^p (\mathfrak{g} , \mathbb{K}) &\to& \Pi^{m+n} C^{m-p}_{int} (\mathfrak{g}, \mathbb{K}) \\
		\nonumber \mathcal{Y}^{* a_1} \wedge \ldots \wedge \mathcal{Y}^{*a_p} &\mapsto& \mathcal{D} \otimes \mathcal{Y}_{a_1} \wedge \ldots \wedge \mathcal{Y}_{a_p} \ ,
	\end{eqnarray}
	such that $\displaystyle \left( \mathcal{D} \otimes \mathcal{Y}_{a_1} \wedge \ldots \wedge \mathcal{Y}_{a_p} \right) \left( \mathcal{Y}^{* b_1} \wedge \ldots \wedge \mathcal{Y}^{*b_p} \right) \propto_{\mathbb{K}} \mathcal{D}$, induces a cohomology isomorphism:
	\begin{equation}\label{IFLSN}
		\star : H^{\bullet} \left( \mathfrak{g} ,  \mathbb{K} \right) \overset{\cong}{\underset{}{\longrightarrow}} \Pi^{m+n} H^{m-\bullet}_{int} \left( \mathfrak{g} , \mathbb{K} \right) \ .
	\end{equation}
\end{prop}

These definitions and results translate to the distributional realisation with the dictionary described above. The action of the CE differential on Dirac delta's is simply given by
\begin{equation}\label{IFLSO}
	d \delta \left( \psi^\alpha \right) = \sum_{a,b,c} f^a_{bc} \mathcal{Y}^{*b} \wedge \mathcal{Y}^{*c} \wedge \iota_{\mathcal{Y}_a} \delta \left( \psi^\alpha \right) = \sum_{\alpha,b,c} f^\alpha_{bc} \mathcal{Y}^{*b} \wedge \mathcal{Y}^{*c} \wedge \iota_{\xi_\alpha} \delta \left( \psi^\alpha \right) \ .
\end{equation}
This definition, equipped with the rules listed in \eqref{IFLSF} can be analogously be interpreted using the last identity of \eqref{IFLSF}, namely, the chain rule:
\begin{equation}\label{IFLSP}
	d \delta \left( \psi^\alpha \right) = \left( d \psi^\alpha \right) \wedge \iota_{\xi_\alpha} \delta \left( \psi^\alpha \right) = \left( \sum_{a,b,c} f^a_{bc} \mathcal{Y}^{*b} \wedge \mathcal{Y}^{*c} \iota_{\mathcal{Y}_a} \psi^\alpha \right) \wedge \iota_{\xi_\alpha} \delta \left( \psi^\alpha \right) = $$ $$ = \sum_{b,c} f^\alpha_{bc} \mathcal{Y}^{*b} \wedge \mathcal{Y}^{*c} \wedge \iota_{\xi_\alpha} \delta \left( \psi^\alpha \right) \ .
\end{equation}

\subsection{Pseudoforms as Infinite-Dimensional Representations}\label{Pseudoforms as Infinite-Dimensional Representations}

In the previous section, we introduced pseudoforms as formal objects containing a non-maximal and non-zero number of Dirac deltas. In this section, we anticipate some constructions of the paper to introduce some pseudoforms \emph{independently} of the distributional realisation. In particular, we will introduce them as infinite-dimensional modules associated with sub-superalgebras or super-cosets. We will use the distributional realisation to define a $\mathfrak{g}$-module (and in particular the $\mathfrak{g}$-action) and a nilpotent differential (actually, it was already defined in \eqref{IFLSP}); then, the dictionary introduced in the previous section will allow translating the $\mathfrak{g}$-module structure and the definition of the differential in the usual polyvector realisation of Berezinians.

We start from a Lie (super)algebra $\mathfrak{g}$, $\dim \mathfrak{g} = (m|n)$, a sub-(super)algebra $\mathfrak{h}$, $\dim \mathfrak{h} = (p|q)$ and denote their (super)coset $\mathfrak{k} = \mathfrak{g} / \mathfrak{h}$, $\dim \mathfrak{k} = (m-p|n-q)$. Associated with $\mathfrak{h}$ and $\mathfrak{k}$ we can define the Berezinian spaces $\mathpzc{B}er \left( \mathfrak{h} \right)$ and $\mathpzc{B}er \left( \mathfrak{k} \right)$, which can be locally realised as
\begin{eqnarray*}
	\mathpzc{B}er \left( \mathfrak{h} \right) &\defeq& \mathbb{K} \cdot \left[ \bigwedge_{i=1}^{p} V^i \otimes \bigwedge_{\alpha=1}^{q} \xi_\alpha \right] \equiv \mathbb{K} \cdot \mathcal{D}_{\mathfrak{h}} \ \leftrightsquigarrow \ \mathbb{K} \cdot \left[ \bigwedge_{i=1}^{p} V^i \wedge \bigwedge_{\alpha=1}^{q} \delta \left( \psi^\alpha \right) \right] \ , \ V^i, \psi^\alpha \in \Pi \mathfrak{h}^* \ , \ \xi_\alpha \in \mathfrak{h} \ , \\
	\mathpzc{B}er \left( \mathfrak{k} \right) &\defeq& \mathbb{K} \cdot \left[ \bigwedge_{\hat{i}=1}^{m-p} V^{\hat{i}} \otimes \bigwedge_{\hat{\alpha}=1}^{n-q} \xi_{\hat{\alpha}} \right] \equiv \mathbb{K} \cdot \mathcal{D}_{\mathfrak{k}}  \ \leftrightsquigarrow \ \mathbb{K} \cdot \left[ \bigwedge_{\hat{i}=1}^{m-p} V^{\hat{i}} \wedge \bigwedge_{\hat{\alpha}=1}^{n-q} \delta \left( \psi^{\hat{\alpha}} \right) \right]\ , \ V^{\hat{i}}, \psi^{\hat{\alpha}} \in \Pi \mathfrak{k}^* \ , \ \xi_{\hat{\alpha}} \in \mathfrak{k} \ .
\end{eqnarray*}
We can now define a $\mathfrak{g}$-action on the two spaces above and we will see that $\mathpzc{B}er \left( \mathfrak{h} \right)$ and $\mathpzc{B}er \left( \mathfrak{k} \right)$ \emph{are not $\mathfrak{g}$-modules}. Nonetheless, we will infer how we should complement these two spaces to be (infinite-dimensional) $\mathfrak{g}$-modules. For that, we use the distributional realisation, the translation to the polyvector realisation is straightforward, given the dictionary described above. In particular, given $\mathcal{Y}_a \in \mathfrak{g}$, we define the $\mathfrak{g}$-action on $\mathpzc{B}er \left( \bullet \right) $ as a (graded) Lie derivative:
\begin{eqnarray}
	\nonumber \mathcal{L} : \mathfrak{g} \times \mathpzc{B}er \left( \bullet \right) &\to& X \\
	\label{PIDRA} \left( \mathcal{Y}_a , \mathcal{D}_{\bullet} \right) &\mapsto& \mathcal{L} \left( \mathcal{Y}_a \right) \mathcal{D}_\bullet \equiv \mathcal{L}_{\mathcal{Y}_a} \mathcal{D}_{\bullet} = \left( d \iota_{\mathcal{Y}_a} + \left( -1 \right)^{\left| \mathcal{Y}_a \right|} \iota_{\mathcal{Y}_a} d \right) \mathcal{D}_{\bullet} \ ,
\end{eqnarray}
where we have left the target space $X$ unspecified on purpose. By directly inspecting the action in \eqref{PIDRA}, we have
\begin{eqnarray*}
	\mathcal{Y}_a \leftrightsquigarrow X_i \ &\implies& \ \mathcal{L}_{X_i} \mathcal{D}_{\bullet} = \left( d \iota_{X_i} + \iota_{X_i} d \right) \mathcal{D}_{\bullet} = \sum_{j,k} 2 f^j_{ik} V^k \iota_{X_j} \mathcal{D}_{\bullet} + \sum_{\alpha, \beta} f^\alpha_{i \beta} \psi^\beta \iota_{\xi_\alpha} \mathcal{D}_{\bullet} \ , \\
	\mathcal{Y}_a \leftrightsquigarrow \xi_\alpha \ &\implies& \ \mathcal{L}_{\xi_\alpha} \mathcal{D}_{\bullet} = \left( d \iota_{\xi_\alpha} - \iota_{\xi_\alpha} d \right) \mathcal{D}_{\bullet} = \sum_{i,\beta} 2 f^i_{\alpha \beta} \psi^\beta \iota_{X_i} \mathcal{D}_{\bullet} + \sum_{i, \beta} f^\beta_{i \alpha} V^i \iota_{\xi_\beta} \mathcal{D}_{\bullet} \ .
\end{eqnarray*}
First of all, note that in general the action of the differential on $\mathcal{D}_\bullet$ is not trivial, as opposed to what happens when considering integral forms as in \eqref{IFLSL}. Secondly, note that that the Lie derivative produces not only objects in $\mathpzc{B}er \left( \bullet \right)$, but also objects in $\Pi \mathfrak{g} \otimes \mathpzc{B}er \left( \bullet \right) \otimes \Pi \mathfrak{g}^*$. More explicitly, we have determined the target space $X$ as
\begin{eqnarray*}
	\mathcal{L} : \mathfrak{g} \times \mathpzc{B}er \left( \mathfrak{h} \right) &\to& \mathpzc{B}er \left( \mathfrak{h} \right) \oplus \left[ \left( \Pi \mathfrak{h} \otimes \mathpzc{B}er \left( \mathfrak{h} \right) \right) \otimes \Pi \mathfrak{k}^* \right] \ , \\
	\mathcal{L} : \mathfrak{g} \times \mathpzc{B}er \left( \mathfrak{k} \right) &\to& \mathpzc{B}er \left( \mathfrak{k} \right) \oplus \left[ \left( \Pi \mathfrak{k} \otimes \mathpzc{B}er \left( \mathfrak{k} \right) \right) \otimes \Pi \mathfrak{h}^* \right] \ .
\end{eqnarray*}
We can define two $\mathfrak{g}$-modules, which we call the modules of $(p|q)$- and $(m-p|n-q)$- pseudoforms as follows.
\begin{defn}[Pseudoform Modules Relative to a Lie sub-(super)algebra]

	Given a Lie superalgebra $\mathfrak{g}$ of dimension $\dim \left( \mathfrak{g} \right) = (m|n)$, a Lie sub-(super)algebra $\mathfrak{h}$ of dimension $\dim \left( \mathfrak{h} \right) = (p|q)$, and denoting $\mathfrak{k} = \mathfrak{g}/\mathfrak{h}$, we define the \emph{modules of $(p|q)$- and $(m-p|n-q)$- pseudoforms} as
	\begin{eqnarray}
		\label{PIDRB} V_{\mathfrak{h}}^{(p|q)} &\defeq& \bigoplus_{i=0}^\infty \left( S^i \Pi \mathfrak{h} \otimes \mathpzc{B}er \left( \mathfrak{h} \right) \right) \otimes S^i \Pi \mathfrak{k}^* = \bigoplus_{i=0}^\infty C^{m-i}_{int} \left( \mathfrak{h} \right) \otimes C^i \left( \mathfrak{k} \right) \ , \\
		\label{PIDRC} V_{\mathfrak{k}}^{(m-p|n-q)} &\defeq& \bigoplus_{i=0}^\infty \left( S^i \Pi \mathfrak{k} \otimes \mathpzc{B}er \left( \mathfrak{k} \right) \right) \otimes S^i \Pi \mathfrak{h}^* = \bigoplus_{i=0}^\infty C^{m-p-i}_{int} \left( \mathfrak{k} \right) \otimes C^i \left( \mathfrak{h} \right) \ .
	\end{eqnarray}
	In particular, pseudoforms are constructed as integral forms of $\mathfrak{h}$ (resp., $\mathfrak{k}$) tensored with superforms of $\mathfrak{k}$ (resp., $\mathfrak{h}$). \eqref{PIDRB} and \eqref{PIDRC} define $\mathfrak{g}$-modules, where the $\mathfrak{g}$-action is given by
	\begin{eqnarray}
	\nonumber \mathcal{L} : \mathfrak{g} \times V_{\bullet}^{(\bullet|\bullet)} &\to& V_{\bullet}^{(\bullet|\bullet)} \\
	\label{PIDRD} \left( \mathcal{Y}_a , \mathcal{D}_{\bullet} \right) &\mapsto& \mathcal{L} \left( \mathcal{Y}_a \right) \mathcal{D}_\bullet \equiv \mathcal{L}_{\mathcal{Y}_a} \mathcal{D}_{\bullet} = \left( d \iota_{\mathcal{Y}_a} + \left( -1 \right)^{\left| \mathcal{Y}_a \right|} \iota_{\mathcal{Y}_a} d \right) \mathcal{D}_{\bullet} \ ,
	\end{eqnarray}
	where $\mathcal{L} : \mathfrak{g} \to End \left( V_{\bullet}^{(\bullet|\bullet)} \right)$ by construction satisfies
	\begin{equation}\label{PIDRE}
		\mathcal{L} \left( \left[ \mathcal{Y}_a , \mathcal{Y}_b \right] \right) = \mathcal{L} \left( \mathcal{Y}_a \right) \mathcal{L} \left( \mathcal{Y}_b \right) - (-1)^{\left| \mathcal{Y}_a \right| \left| \mathcal{Y}_b \right|} \mathcal{L} \left( \mathcal{Y}_b \right) \mathcal{L} \left( \mathcal{Y}_a \right) \ .
	\end{equation}
\end{defn}

Starting from the two spaces \eqref{PIDRB} and \eqref{PIDRC}, we can now define \emph{pseudoform cochains} and \emph{pseudoform cohomology}, with respect to a given sub-(super)algebra:
\begin{defn}[Pseudoform Cohomology Relative to a sub-(super)algebra]\label{Pseudoform Cohomology Relative to a sub-(super)algebra}
	
	Given a Lie superalgebra $\mathfrak{g}$ of dimension $\dim \left( \mathfrak{g} \right) = (m|n)$, a Lie sub-(super)algebra $\mathfrak{h}$ of dimension $\dim \left( \mathfrak{h} \right) = (p|q)$, and denoting $\mathfrak{k} = \mathfrak{g}/\mathfrak{h}$, and given $s \in \mathbb{N} \cup \left\lbrace 0 \right\rbrace$, we define the spaces of \emph{$(p+s|q)$-, $(p-s|q)$-, $(m-p+s|n-q)$- and $(m-p-s|n-q)$- cochains as}
	\begin{eqnarray}
		\label{PIDRF} C^{p+s} \left( \mathfrak{g}, V_{\mathfrak{h}}^{(p|q)} \right) &\defeq& C^s \left( \mathfrak{k} \right) \otimes V_{\mathfrak{h}}^{(p|q)} = \bigoplus_{i=0}^\infty C^{p-i}_{int} \left( \mathfrak{h} \right) \otimes C^{i+s} \left( \mathfrak{k} \right) \ , \\
		\label{PIDRG} C^{p-s} \left( \mathfrak{g}, V_{\mathfrak{h}}^{(p|q)} \right) &\defeq& C_s \left( \mathfrak{h} \right) \otimes V_{\mathfrak{h}}^{(p|q)} = \bigoplus_{i=0}^\infty C^{p-s-i}_{int} \left( \mathfrak{h} \right) \otimes C^{i} \left( \mathfrak{k} \right) \ , \\
		\label{PIDRH} C^{m-p+s} \left( \mathfrak{g}, V_{\mathfrak{k}}^{(m-p|n-q)} \right) &\defeq& C^s \left( \mathfrak{h} \right) \otimes V_{\mathfrak{k}}^{(m-p|n-q)} = \bigoplus_{i=0}^\infty C^{m-p-i}_{int} \left( \mathfrak{k} \right) \otimes C^{i+s} \left( \mathfrak{h} \right) \ , \\
		\label{PIDRI} C^{m-p-s} \left( \mathfrak{g}, V_{\mathfrak{k}}^{(m-p|n-q)} \right) &\defeq& C_s \left( \mathfrak{k} \right) \otimes V_{\mathfrak{k}}^{(m-p|n-q)} = \bigoplus_{i=0}^\infty C^{m-p-s-i}_{int} \left( \mathfrak{k} \right) \otimes C^{i} \left( \mathfrak{h} \right) \ .
	\end{eqnarray}
	We will adopt the shorter notation
	\begin{equation}\label{PIDRJ}
		C^{p \pm s} \left( \mathfrak{g}, V_{\mathfrak{h}}^{(p|q)} \right) \equiv C^{(p \pm s|q)} \left( \mathfrak{g} \right) \ , \ C^{m -p \pm s} \left( \mathfrak{g}, V_{\mathfrak{k}}^{(m-p|n-q)} \right) \equiv C^{(m-p \pm s|n-q)} \left( \mathfrak{g} \right) \ ,
	\end{equation}
	where the second number in the apex (which we called \emph{picture number} in the previous section), indicates on which module they take values. We can lift $C^{(s|q)} \left( \mathfrak{g} \right)$ and $C^{(s|n-q)} \left( \mathfrak{g} \right)$ to cochain complexes by introducing a Chevalley-Eilenberg differential; in the distributional realisation, this is simply given by
	\begin{eqnarray}
		\label{PIDRK} d : C^{(s|\bullet)} \left( \mathfrak{g} \right) &\to& C^{(s+1|\bullet)} \left( \mathfrak{g} \right) \\
		\nonumber \omega &\mapsto& d \omega = \sum_{a,b,c} f^a_{bc} \mathcal{Y}^{*b} \wedge \mathcal{Y}^{* c} \wedge \iota_{\mathcal{Y}_c} \omega \ .
	\end{eqnarray}
	The translation to polyvectors realisation is again straightforward, but one should use different differentials when acting on forms or vectors, thus making the expressions more cumbersome.
	
	We can now define closed and exact pseudoforms as in \eqref{PE} and \eqref{PF}, where we take as module $V$ either $V_{\mathfrak{h}}^{(p|q)}$ or $V_{\mathfrak{k}}^{(m-p|n-q)}$ defined in \eqref{PIDRB} and \eqref{PIDRC}, respectively:
	\begin{eqnarray}
		\label{PIDRL} H^{(\bullet|q)} \left( \mathfrak{g} \right) &\defeq& H^\bullet \left( \mathfrak{g} , V_{\mathfrak{h}}^{(p|q)} \right) \ , \\
		\label{PIDRM} H^{(\bullet|n-q)} \left( \mathfrak{g} \right) &\defeq& H^\bullet \left( \mathfrak{g} , V_{\mathfrak{k}}^{(m-p|n-q)} \right) \ .
	\end{eqnarray}
	The previous definitions can be naturally extended if we consider also a $\mathfrak{g}$-module $V$: we will define the pseudoform cochains with values in the module $V$ as
	\begin{equation}\label{PIDRN}
		C^{(\bullet|q)} \left( \mathfrak{g} , V \right) \defeq C^\bullet \left( \mathfrak{g} , V_{\mathfrak{h}}^{(p|q)} \otimes V \right) \ , \ C^{(\bullet|n-q)} \left( \mathfrak{g} , V \right) \defeq C^\bullet \left( \mathfrak{g} , V_{\mathfrak{k}}^{(m-p|n-q)} \otimes V \right) \ ,
	\end{equation}
	and the cohomology will be defined as usual with respect to the modules $V_{\mathfrak{h}}^{(p|q)} \otimes V$ and $V_{\mathfrak{k}}^{(m-p|n-q)} \otimes V$.
\end{defn}

In this section, we have given the definitions of pseudoforms and cohomology of pseudoforms in the algebraic context. In the following sections, we will show how the Koszul-Hochschild-Serre spectral sequence can be extended to these new objects. In particular, by its very construction, we will see that the filtrations (as we will explain in the following section we will use two inequivalent filtrations associated with a given sub-(super)algebra) used to introduce the spectral sequence are made of finite-dimensional spaces, even when dealing with pseudoforms, thus greatly simplifying the task. We will use an explicit example as a guide and then extract general results that allow us to extend Thm.\ref{KHStheorem}.

\section{Cohomology via Spectral Sequences: $\mathfrak{osp}(2|2)$}\label{Cohomology via Spectral Sequences}

In this section, we generalize the Koszul spectral sequence to compute pseudoform cohomology of the Lie superalgebra $\mathfrak{osp}(2|2)$ via the associated supercoset $\mathfrak{osp}(2|2) / \mathfrak{osp}(1|2)$. We will use this guiding example to infer some general constructions which will be introduced in the next section.

By denoting with $Z,H,E^{\pm\pm}$ the bosonic generators and with $F^\pm, \bar{F}^\pm$ the fermionic ones, the non trivial (anti-)commutation relations of $\mathfrak{osp}(2|2)$ (see e.g., \cite{Frappat}\footnote{With respect to the fermionic generators of \cite{Frappat}, we used the combinations $\displaystyle \frac{1}{2} \left( F^\pm \pm \bar{F}^\pm \right) $ and $\displaystyle \frac{1}{2} \left( F^\pm \mp \bar{F}^\pm \right) $ in order to make the $\mathfrak{osp}(1|2)$ sub-superalgebra manifest.}) are
\begin{eqnarray}
	\label{VSPA} && \left[ H, E^{\pm\pm} \right] = \pm E^{\pm\pm} \ , \ \left[ E^{++} , E^{--} \right] = 2 H \ , \ \left[ H , F^\pm \right] = \pm \frac{1}{2} F^\pm \ , \\
	\label{VSPB} && \left[ E^{\pm\pm} , F^\mp \right] = - F^\pm \ , \ \left\lbrace F^\pm , F^\pm \right\rbrace = \pm \frac{1}{2} E^{\pm\pm} \ , \ \left\lbrace F^+ , F^- \right\rbrace = H \ , \\
	\label{VSPC} && \left[ Z, \bar{F}^\pm \right] = \frac{1}{2} F^\pm \ , \ \left\lbrace \bar{F}^\pm , \bar{F}^\pm \right\rbrace = \mp \frac{1}{2} E^{\pm\pm} \ , \\
	\label{VSPD} && \left[ Z , F^\pm \right] = \frac{1}{2} \bar{F}^\pm \ , \ \left[ H , \bar{F}^\pm \right] = \pm \frac{1}{2} \bar{F}^\pm \ , \ \left[ E^{\pm\pm} , \bar{F}^\mp \right] = - \bar{F}^\pm \ , \ \left\lbrace \bar{F}^\pm , F^\mp \right\rbrace = \mp \frac{1}{2} Z \ .
\end{eqnarray}
By denoting with $\mathfrak{g} = \mathfrak{osp}(2|2)$, $\mathfrak{h}= \mathfrak{osp}(1|2)$ and $\mathfrak{k}= \mathfrak{g}/\mathfrak{h}$ (again, as a quotient of vector spaces as in Def.\ref{Relative Lie (super)algebra cohomology}), we notice that the lines 
\eqref{VSPA} and \eqref{VSPB} represent the sub-(super)algebra $\mathfrak{h}$ and, schematically, we have
\begin{equation}\label{VSPE}
	\left[ \mathfrak{h} , \mathfrak{h} \right] \subseteq \mathfrak{h} \ , \ \left[ \mathfrak{h} , \mathfrak{k} \right] \subseteq \mathfrak{k} \ , \ \left[ \mathfrak{k} , \mathfrak{k} \right] \subseteq \mathfrak{h} \ .
\end{equation}
Then, the sub-(super)algebra $\mathfrak{h}$ is \emph{reductive} in $\mathfrak{g}$ and the (super)coset $\mathfrak{k}$ is \emph{homogeneous}. 

We can now move on to the dual forms in $\Pi \mathfrak{g}^*$ by translating the commutation relations into MC equations. We denote by $\mathfrak{g}^*, 
\mathfrak{h}^*, \mathfrak{k}^*$ the dual spaces (the parity inversion is understood) to $\mathfrak{g}, \mathfrak{h}, \mathfrak{k}$: given a basis vector $\mathcal{Y}_a$, its dual $\mathcal{Y}^{*a}$ satisfies $\iota_{\mathcal{Y}_a} \mathcal{Y}^{*b} =\delta_a^b$.

The forms dual to the generators are denoted as
\begin{equation}\label{VSPF}
	H \leftrightsquigarrow V^0 \ , \ E^{\pm\pm} \leftrightsquigarrow V^{\pm\pm} \ , \ Z \leftrightsquigarrow  U \ , \ F^\pm \leftrightsquigarrow \psi^\pm \ , \ \bar{F}^\pm \leftrightsquigarrow \bar{\psi}^\pm \ ,
\end{equation}
and the MC equations read
\begin{eqnarray}
	\label{VSPG} d V^0 &=& 2 V^{++} \wedge V^{--} + \psi^+ \wedge \psi^- \ , \\
	\label{VSPH} d V^{\pm\pm} &=& \pm V^0 \wedge V^{\pm\pm} \pm \frac{1}{2} \left( \psi^\pm\wedge  \psi^\pm\right) \mp \frac{1}{2} \left( \bar{\psi}^\pm \wedge \bar{\psi}^\pm \right) \ , \\
	\label{VSPI} d \psi^\pm &=& \pm \frac{1}{2} V^0 \wedge \psi^\pm - V^{\pm\pm} \wedge \psi^\mp + \frac{1}{2} U \wedge \bar{\psi}^\pm \ , \\
	\label{VSPJ} d U &=& - \frac{1}{2} \bar{\psi}^+ \wedge \psi^- + \frac{1}{2} \bar{\psi}^- \wedge \psi^+ \ , \\
	\label{VSPK} d \bar{\psi}^\pm &=& \frac{1}{2} U \wedge \psi^\pm \pm \frac{1}{2} V^0 \wedge \bar{\psi}^\pm - V^{\pm\pm} \wedge \bar{\psi}^\mp \ .
\end{eqnarray}
When taking the coset w.r.t. the sub-superalgebra $\mathfrak{osp}(1|2)$, generated by $H, E^{\pm \pm}, F^\pm$ as in \eqref{VSPA} and \eqref{VSPB}, the CE differential is modified in a covariant derivative as
\begin{equation}\label{VSPL}
	d \to \nabla_{\mathfrak{k}} = d - A \ ,
\end{equation}
where $A$ is a $3\times 3$ supermatrix whose entries are the $\mathfrak{osp}(1|2)$ forms, acting on the vector $(U, \bar{\psi}^+ , \bar{\psi}^-)$. In particular, we 
can rewrite the MC equations in terms of the curvatures of $\mathfrak{osp}(1|2)$:
\begin{eqnarray}\label{VSPM}
	&&R^0 = 0 \ , \	R^{\pm\pm} = \mp \frac{1}{2} \bar{\psi}^\pm \wedge \bar{\psi}^\pm \ , \ \nabla_{\mathfrak{k}}   \psi^\pm = \frac{1}{2} U \wedge \bar{\psi}^\pm \ , \nonumber \\
	&&\nabla_{\mathfrak{k}} U = 0 \ , \ \nabla_{\mathfrak{k}} \bar{\psi}^\pm = 0 \ .
\end{eqnarray} 
This shows that the operator $\nabla_{\mathfrak{k}}$ has a trivial action on any form in $\mathfrak{osp} \left( 2|2 \right) / \mathfrak{osp}(1|2)$, so that they are (covariantly) closed. 

The cohomology of the supercoset $\mathfrak{k}$ is the \emph{relative cohomology} $H^\bullet \left( \mathfrak{g} , \mathfrak{h} \right)$ as introduced in Def.\ref{Relative Lie (super)algebra cohomology}.

Going back to the example under examination, we have to look for invariant forms. Since every form in $\mathfrak{k}$ is closed, the invariant ones will automatically generate the relative cohomology. It is easy to prove that there are no invariants among superforms, except for the constants:
\begin{equation}\label{VSPP}
	H^p \left( \mathfrak{g} , \mathfrak{h} \right) = \begin{cases}
		\mathbb{K} \ , \ \text{ if } p = 0 \ , \\
		\left\lbrace 0 \right\rbrace \ , \ \text{ else.}
	\end{cases}
\end{equation}
The integral form cohomology of $\mathfrak{k}$ is easily computed by using the isomorphism introduced in Prop.\ref{Berezinian Complement Quasi-Isomorphism} (which trivially extends to the (super) coset):
\begin{equation}\label{VSPQ}
	\star : H^{\bullet} \left( \mathfrak{g} , \mathfrak{h} \right) \overset{\cong}{\underset{}{\longrightarrow}} \Pi H^{1-\bullet}_{integral} \left( \mathfrak{g} , \mathfrak{h} \right) \equiv \Pi H^{(1-\bullet|2)} \left( \mathfrak{g} , \mathfrak{h} \right) \ .
\end{equation}
In particular, this leads to
\begin{equation}\label{VSPR}
	H^p_{integral} \left( \mathfrak{g} , \mathfrak{h} \right) \equiv H^{(p|2)} \left( \mathfrak{g} , \mathfrak{h} \right) = \begin{cases}
		\Pi \mathbb{K} \ , \ \text{ if } p = 1 \ , \\
		\left\lbrace 0 \right\rbrace \ , \ \text{ else.}
	\end{cases}
\end{equation}
The class $\displaystyle \left[ \mathpzc{B}er_{\mathfrak{k}} \right] \in H^1_{integral} \left( \mathfrak{g} , \mathfrak{h} \right)$ corresponds to the Berezinian class of the supercoset $\mathfrak{k}$.

\subsection{The Berezinian $\mathcal{D}_{\mathfrak{k}}$: the Distributional Realisation}

We can check that $\displaystyle \mathcal{D}_{\mathfrak{k}} \in \mathpzc{B}er \left( \mathfrak{k} \right)$ is invariant by using the realisation of integral forms via Dirac deltas.

A representative of the integral form in the class $\displaystyle \mathcal{D}_{\mathfrak{k}}$ is
\begin{equation}\label{AECA}
	\mathcal{D}_{\mathfrak{k}} = U \wedge \delta \left( \bar{\psi}^+ \right) \wedge \delta \left( \bar{\psi}^- \right) \ .
\end{equation}
In order to verify that $\mathcal{D}_{\mathfrak{k}}$ is invariant w.r.t. the sub-superalgebra $\mathfrak{h}$, we have to verify that
\begin{equation}\label{AECB}
	\iota_{\mathcal{Y}_a} \mathcal{D}_{\mathfrak{k}} = 0 \ , \ \mathcal{L}_{\mathcal{Y}_a} \mathcal{D}_{\mathfrak{k}} = 0 \ , \ \forall \mathcal{Y}_a \in \mathfrak{h} \ .
\end{equation}
While the first condition is trivially satisfied, the second one follows from $\displaystyle d \mathcal{D}_{\mathfrak{k}} = 0$ or, explicitly,
\begin{eqnarray}
	 \nonumber d \left[ U \wedge \delta \left( \bar{\psi}^+ \right) \wedge \delta \left( \bar{\psi}^- \right) \right] &=& - \frac12 U \wedge V^0 \wedge \bar{\psi}^+ \wedge \bar{\iota}_+ \delta \left( \bar{\psi}^+ \right) \wedge \delta \left( \bar{\psi}^- \right) + \\
	\nonumber &-& \frac{1}{2} U \wedge \delta \left( \bar{\psi}^+ \right) \wedge V^0 \wedge \bar{\psi}^- \wedge \bar{\iota}_- \delta \left( \bar{\psi}^- \right) = \\
	\label{AECC} &=& \frac{1}{2} U \wedge V^0 \wedge \delta \left( \bar{\psi}^+ \right) \wedge \delta \left( \bar{\psi}^- \right) + \frac{1}{2} U \wedge \delta \left( \bar{\psi}^+ \right) \wedge V^0 \wedge \delta \left( \bar{\psi}^- \right) = 0 \ ,
\end{eqnarray}
where we used the properties in \eqref{IFLSF} and the action of CE differential on this realisation of integral forms defined in \eqref{IFLSP}. Hence, $\mathcal{L}_{\mathcal{Y}_a} \mathcal{D}_{\mathfrak{k}} = d \iota_{\mathcal{Y}_a} \mathcal{D}_{\mathfrak{k}} + (-1)^{|\mathcal{Y}_a|} \iota_{\mathcal{Y}_a} d \mathcal{D}_{\mathfrak{k}} = 0 , \forall \mathcal{Y}_a \in \mathfrak{h}$.

With this realisation, we can write the Berezinian $\displaystyle \mathcal{D}_{\mathfrak{g}} \in \mathpzc{B}er \left( \mathfrak{g} \right)$ of the superalgebra ${\mathfrak{g}}$, which represents the top integral form, as
\begin{eqnarray}
\label{newBERA}
\mathcal{D}_{\mathfrak{g}} = \mathcal{D}_{\mathfrak{k}} \wedge \mathcal{D}_{\mathfrak{h}} = U  \delta \left( \bar{\psi}^+ \right)  \delta \left( \bar{\psi}^- \right)  \wedge    V^0 V^{++} V^{--}  \delta \left({\psi}^+ \right) \delta \left({\psi}^- \right) \ ,
\end{eqnarray}
where the top form $\mathcal{D}_{\mathfrak{k}}$ of the coset \eqref{AECA} multiplies the top form $ \mathcal{D}_{\mathfrak{h}}$ of the Lie sub-(super)algebra. 
$\mathcal{D}_{\mathfrak{g}}$ is a top form in $C^4_{int} \left( \mathfrak{g} \right) \equiv C^{(4|4)} \left( \mathfrak{g} \right)$, then it is closed and not exact. Moreover, notice that, since ${\mathfrak{h}}$ is a Lie sub-(super)algebra, the integral form $ \mathpzc{B}er_{\mathfrak{h}}$, i.e., its Berezinian top form, is an element of the cohomology $H^\bullet_{int} \left( \mathfrak{h} \right)$ as a consequence of the duality in Prop.\ref{Berezinian Complement Quasi-Isomorphism}.

\subsection{The Spectral Sequence}

Let us now focus on calculating the pseudoform cohomology by means of Koszul spectral sequences. As discussed in Sec.\ref{Pseudoforms as Infinite-Dimensional Representations}, the choice of the sub-(super)algebra $\mathfrak{h}$ in $\mathfrak{g}$ induces pseudoforms as integral forms in $\mathfrak{h}$ and $\mathfrak{k}$ or, in other terms, at picture number $\dim \mathfrak{h}_1$ and $\dim \mathfrak{k}_1$. In the specific example under examination, we have pseudoforms at picture number 2 both from $\mathfrak{h}$ and from $\mathfrak{k}$. To construct page 0 of the spectral sequence we consider the filtration defined in Def.\ref{Koszul-Hochschild-Serre Spectral Sequence}, but considering picture-number 2 forms:
\begin{equation}\label{TSSA}
	F^p C^{(q|2)} \left( \mathfrak{g} \right) = \left\lbrace \omega \in C^{(q|2)} \left( \mathfrak{g} \right) : \forall \mathcal{Y}_a \in \mathfrak{h} , \iota_{\mathcal{Y}_{a_1}} \ldots \iota_{\mathcal{Y}_{a_{q+1-p}}} \omega = 0 \right\rbrace \ .
\end{equation}
It is not difficult to verify that
\begin{align}
	\label{TSSB} F^{q+n} C^{(q|2)} \left( \mathfrak{g} \right) &= 0 \ , \ \forall n \in \mathbb{N} \setminus \left\lbrace 0 \right\rbrace , q \in \mathbb{Z} \ , \\
	\label{TSSC} F^{p+1} C^{(q|2)} \left( \mathfrak{g} \right) &\subseteq F^{p} C^{(q|2)} \left( \mathfrak{g} \right) \ , \ \forall p,q \in \mathbb{Z} \ .
\end{align}
In order to verify that \eqref{TSSA} correctly defines a filtration, we have to check that
\begin{equation}\label{TSSD}
	d F^p C^{(q|2)} \left( \mathfrak{g} \right) \subseteq F^p C^{(q+1|2)} \left( \mathfrak{g} \right) \ , \ \forall p,q \in \mathbb{Z} \ .
\end{equation}
Let us verify this for $p=q$ (the generalisation is straightforward). We have
\begin{eqnarray}
	\label{TSSE} F^q C^{(q|2)} \left( \mathfrak{g} \right) &=& \left\lbrace \omega \in C^{(q|2)} \left( \mathfrak{g} \right) : \forall \mathcal{Y}_a \in \mathfrak{h} , \iota_{\mathcal{Y}_a} \omega = 0 \right\rbrace \ , \\
	\label{TSSEA} F^q C^{(q+1|2)} \left( \mathfrak{g} \right) &=& \left\lbrace \omega \in C^{(q+1|2)} \left( \mathfrak{g} \right) : \forall \mathcal{Y}_a \in \mathfrak{h} , \iota_{\mathcal{Y}_{a_1}} \iota_{\mathcal{Y}_{a_2}} \omega = 0 \right\rbrace \ .
\end{eqnarray}
Thus, $d \omega \in F^q C^{(q+1|2)} \left( \mathfrak{g} \right)$ iff $\iota_{\mathcal{Y}_{a_1}} \iota_{\mathcal{Y}_{a_2}} d\omega = 0, \forall \mathcal{Y}_{a_1} , \mathcal{Y}_{a_2} \in \mathfrak{h}$. In particular, we have
\begin{equation}\label{TSSF}
	\iota_{\mathcal{Y}_{a_1}} \iota_{\mathcal{Y}_{a_2}} d\omega = \iota_{\mathcal{Y}_{a_1}} \iota_{\mathcal{Y}_{a_2}} d\omega + \left( -1 \right)^{|\xi_2|+1} \iota_{\mathcal{Y}_{a_1}} d \iota_{\mathcal{Y}_{a_2}} \omega = $$ $$ = \iota_{\mathcal{Y}_{a_1}} \mathcal{L}_{\mathcal{Y}_{a_2}} \omega = \iota_{\mathcal{Y}_{a_1}} \mathcal{L}_{\mathcal{Y}_{a_2}} \omega + \left( -1 \right)^{|\xi_1||\xi_2|+1} \mathcal{L}_{\mathcal{Y}_{a_2}} \iota_{\mathcal{Y}_{a_1}} \omega = \iota_{\left[ \mathcal{Y}_{a_1} , \mathcal{Y}_{a_2} \right]} \omega = 0 \ ,
\end{equation}
where we systematically added trivial terms and used the definitions \eqref{TSSE} and \eqref{TSSEA}. The extension to any $p$ leads to the same type of manipulations. Hence we verified that \eqref{TSSA} correctly defines a filtration on $C^{(\bullet|2)} \left( \mathfrak{g} \right)$.

There are major differences between the spectral sequence used for conventional Lie algebras and the one used for pseudoforms for Lie superalgebras. For the former we have $q \in \left\lbrace 0 , 1 , \ldots , \text{dim} \, \mathfrak{g} \right\rbrace$, for the latter $q \in \mathbb{Z}$. This is a consequence of the fact that the complex of superforms is unbounded from above and the complex of integral forms is unbounded from below. 
Since the pseudoforms induced from the filtration \eqref{TSSA} arise as (see \eqref{PIDRF} $\div$ \eqref{PIDRI})
\begin{equation}\label{TSSFA}
	C^{(\bullet|2)} \left( \mathfrak{g} \right) = \bigoplus C_{int}^\bullet \left( \mathfrak{k} \right) \otimes C^\bullet \left( \mathfrak{h} \right) \ ,
\end{equation}
the complex is unbounded both from above and from below. 

In addition, for Lie algebras $\mathfrak{g}$, one always has
	\begin{equation}\label{TSSG}
		F^p C^q  \left( \mathfrak{g} \right) = C^q \left( \mathfrak{g} \right) \ , \ \forall p \leq 0 \ ,
	\end{equation}
since the contraction operator $\iota_\xi$ is odd for any $\xi \in \mathfrak{g}$. On the contrary, for superalgebras \eqref{TSSG} does not hold.

It is convenient to use the following notation for the spaces of the filtration \eqref{TSSA}
\begin{equation}\label{TSSH}
	F^{p} C^{(q|2)}\left( \mathfrak{g} \right) \equiv C^{(q|2)}_{q-p} \left( \mathfrak{g} \right) \defeq \bigoplus_{i=0}^{q-p} C^{(q-i|2)}(\mathfrak{k}) \otimes  C^{(i|0)}(\mathfrak{h}) \ ,
\end{equation}
i.e., the space $(q|2)$-forms depending \emph{at most} on $(q-p)$-superforms in $\mathfrak{h}^*$. For example, 
\begin{eqnarray}
	\label{AN_B} C^{(q|2)}_{1} \left( \mathfrak{g} \right) &=& C^{(q|2)} \left( \mathfrak{k} \right) \oplus \Big( C^{(q-1|2)} \left( \mathfrak{k} \right) \otimes C^{(1|0)} \left( \mathfrak{h}\right) \Big) \ , \nonumber \\
	\nonumber C^{(q|2)}_{2} \left( \mathfrak{g} \right) &=& C^{(q|2)} \left( \mathfrak{k} \right) \oplus \Big( C^{(q-1|2)} \left( \mathfrak{k} \right) \otimes C^{(1|0)} \left( \mathfrak{h}\right) \Big) \oplus \Big( C^{(q-2|2)} \left( \mathfrak{k} \right) \otimes C^{(2|0)} \left( \mathfrak{h} \right) \Big) \ .
\end{eqnarray}
This is one of the main strengths of this technique to calculate pseudoform cohomology. Even though we start from infinite-dimensional spaces, the filtration introduced in \eqref{TSSA} allows to consider finite-dimensional subspaces of pseudoforms, only. In the present case where $\text{dim} \mathfrak{k} = (1|2)$, namely, when $C^{(q|2)} \left( \mathfrak{k} \right) = \left\lbrace 0 \right\rbrace$ for $q >1$, the definition in \eqref{TSSA} implies the further simplifications
\begin{eqnarray}\label{TSSI}
	&&F^p C^{(q|2)} \left( \mathfrak{g} \right) = 0 \ , \ \forall p \geq 2 \ , \ q \geq 1 \ , \nonumber \\
	&&F^1 C^{(q|2)} \left( \mathfrak{g} \right) = C_{q-1}^{(q|2)} \left( \mathfrak{g} \right) = C^{(1|2)} \left( \mathfrak{k} \right) \otimes C^{(q-1|0)} \left( \mathfrak{h} \right) \ , \ \forall q \geq 1 \ .
\end{eqnarray}

We can now define page 0 of the spectral sequence for the pseudoforms of $\mathfrak{g}$ associated to the filtration \eqref{TSSA} as
\begin{equation}\label{TSSJ}
	E_0^{m,n} \defeq F^m C^{(m+n|2)} \left( \mathfrak{g} \right) / F^{m+1} C^{(m+n|2)} \left( \mathfrak{g} \right) \ .
\end{equation}
In Table \ref{TableTSSB} we collect the whole page at picture number 2 for the example under examination.
\vskip .3cm
\begin{table}[ht!]
\centering
\begin{tabular}{cccccc}
 \multicolumn{1}{c}~                               
&  \multicolumn{1}{c}0                                                                  
&  \multicolumn{1}{c}1                                                                                                                     
&  \multicolumn{1}{c}2                                                                                                                     
&  \multicolumn{1}{c}3                                                                                                                     
& $\ldots$ 
\\ \cline{2-6} 
\multicolumn{1}{c|}{$\vdots$} 
& \multicolumn{1}{c|}{$\ldots$}                                      
& \multicolumn{1}{c|}{$\ldots$}                                                                                         
& \multicolumn{1}{c|}{$\ldots$}                                                                                         
& \multicolumn{1}{c|}{$\ldots$}                                                                                         
& $\ldots$ 
\\ \cline{2-6} 
\multicolumn{1}{c|}{2}        
& \multicolumn{1}{c|}{0}                                             
& \multicolumn{1}{c|}{0}                                                                                                
& \multicolumn{1}{c|}{0}                                                                                                
& \multicolumn{1}{c|}{0}                                                                                                
& $\ldots$ 
\\ \cline{2-6} 
\multicolumn{1}{c|}{1}        
& \multicolumn{1}{c|}{$C^{(1|2)} \left( \mathfrak{k} \right)$}  
& \multicolumn{1}{c|}{$C^{(1|2)} \left( \mathfrak{k} \right) \otimes C^{(1|0)} \left( \mathfrak{h} \right)$}  
& \multicolumn{1}{c|}{$C^{(1|2)} \left( \mathfrak{k} \right) \otimes C^{(2|0)} \left( \mathfrak{h} \right)$}  
& \multicolumn{1}{c|}{$C^{(1|2)} \left( \mathfrak{k} \right) \otimes C^{(3|0)} \left( \mathfrak{h} \right)$}  
& $\ldots$ 
\\ \cline{2-6} 
\multicolumn{1}{c|}{0}        
& \multicolumn{1}{c|}{$C^{(0|2)} \left( \mathfrak{k} \right)$}  
& \multicolumn{1}{c|}{$C^{(0|2)} \left( \mathfrak{k} \right) \otimes C^{(1|0)} \left( \mathfrak{h} \right)$}  
& \multicolumn{1}{c|}{$C^{(0|2)} \left( \mathfrak{k} \right) \otimes C^{(2|0)} \left( \mathfrak{h} \right)$}  
& \multicolumn{1}{c|}{$C^{(0|2)} \left( \mathfrak{k} \right) \otimes C^{(3|0)} \left( \mathfrak{h} \right)$}  
& $\ldots$ 
\\ \cline{2-6} 
\multicolumn{1}{c|}{-1}       
& \multicolumn{1}{c|}{$C^{(-1|2)} \left( \mathfrak{k} \right)$} 
& \multicolumn{1}{c|}{$C^{(-1|2)} \left( \mathfrak{k} \right) \otimes C^{(1|0)} \left( \mathfrak{h} \right)$} 
& \multicolumn{1}{c|}{$C^{(-1|2)} \left( \mathfrak{k} \right) \otimes C^{(2|0)} \left( \mathfrak{h} \right)$} 
& \multicolumn{1}{c|}{$C^{(-1|2)} \left( \mathfrak{k} \right) \otimes C^{(3|0)} \left( \mathfrak{h} \right)$} 
& $\ldots$ 
\\ \cline{2-6} 
\multicolumn{1}{c|}{-2}       
& \multicolumn{1}{c|}{$C^{(-2|2)} \left( \mathfrak{k} \right)$} 
& \multicolumn{1}{c|}{$C^{(-2|2)} \left( \mathfrak{k} \right) \otimes C^{(1|0)} \left( \mathfrak{h} \right)$} 
& \multicolumn{1}{c|}{$C^{(-2|2)} \left( \mathfrak{k} \right) \otimes C^{(2|0)} \left( \mathfrak{h} \right)$} 
& \multicolumn{1}{c|}{$C^{(-2|2)} \left( \mathfrak{k} \right) \otimes C^{(3|0)} \left( \mathfrak{h} \right)$} 
& $\ldots$ 
\\ \cline{2-6} 
\multicolumn{1}{c|}{-3}       
& \multicolumn{1}{c|}{$C^{(-3|2)} \left( \mathfrak{k} \right)$} 
& \multicolumn{1}{c|}{$C^{(-3|2)} \left( \mathfrak{k} \right) \otimes C^{(1|0)} \left( \mathfrak{h} \right)$} 
& \multicolumn{1}{c|}{$C^{(-3|2)} \left( \mathfrak{k} \right) \otimes C^{(2|0)} \left( \mathfrak{h} \right)$} 
& \multicolumn{1}{c|}{$C^{(-3|2)} \left( \mathfrak{k} \right) \otimes C^{(3|0)} \left( \mathfrak{h} \right)$} 
& $\ldots$ 
\\ \cline{2-6} 
\multicolumn{1}{c|}{$\vdots$} 
& \multicolumn{1}{c|}{$\ldots$}                                      
& \multicolumn{1}{c|}{$\ldots$}                                                                                         
& \multicolumn{1}{c|}{$\ldots$}                                                                                         
& \multicolumn{1}{c|}{$\ldots$}                                                                                         
& $\ldots$
\end{tabular}
\caption{$E_0^{m,n}$ as defined in \eqref{TSSJ}. The integers $n$ and $m$ are spanned on the horizontal and vertical axes, respectively.}\label{TableTSSB}
\end{table}

In order to proceed with the construction of the spectral sequence, we have to define the differentials to move from one page to the other. First of all, from the MC equations (\ref{VSPG}) 
$\div$ (\ref{VSPK}), we can schematically write the CE differential for $\mathfrak{g} = \mathfrak{osp} (2|2)$ and $\mathfrak{h} = \mathfrak{osp} (1|2)$ as
\begin{equation}\label{TSSK}
	d = \mathcal{Y}^*_{\mathfrak{h}} \mathcal{Y}^*_{\mathfrak{h}} \iota_{\mathcal{Y}_\mathfrak{h}} + \mathcal{Y}^*_{\mathfrak{k}} \mathcal{Y}^*_{\mathfrak{k}} \iota_{\mathcal{Y}_\mathfrak{h}} +\mathcal{Y}^*_{\mathfrak{k}} \mathcal{Y}^*_{\mathfrak{h}} \iota_{\mathcal{Y}_\mathfrak{k}} \ .
\end{equation}
The first differential $d_0$ of the spectral sequence is the \emph{horizontal differential} induced by $d$, as indicated in \eqref{KHSSSD}; it acts on objects in \eqref{TSSJ} as
\begin{eqnarray}\label{TSSL}
	d_0 : E_0^{m,n} &\longrightarrow& E_0^{m,n+1} \nonumber \\
	\omega &\mapsto& d_0 \omega \defeq \left( \mathcal{Y}^*_{\mathfrak{k}} \mathcal{Y}^*_{\mathfrak{h}} \iota_{\mathcal{Y}_\mathfrak{k}} + \mathcal{Y}^*_{\mathfrak{h}} \mathcal{Y}^*_{\mathfrak{h}} \iota_{\mathcal{Y}_\mathfrak{h}} \right) \omega \ \ .
\end{eqnarray}
Page 1 of the spectral sequence is defined as the cohomology of page 0 with respect to  the differential $d_0$:
\begin{equation}\label{TSSM}
	E_1^{m,n} \defeq H \left( E_0^{m,n} , d_0 \right) \ .
\end{equation}
We can explicitly calculate page 1 in the case under consideration: given $d_0$ as in \eqref{TSSL} (in particular, because of the reductivity), we can treat the two factors in the tensor product of $E_0^{m,n}$ separately. We then have
\begin{equation}\label{TSSN}
	H \left( C^{(m|2)} \left( \mathfrak{k} \right) \otimes C^{(n|0)} \left( \mathfrak{h} \right) , \mathcal{Y}^*_{\mathfrak{h}} \mathcal{Y}^*_{\mathfrak{h}} \iota_{\mathcal{Y}_\mathfrak{h}} \right) = C^{(m|2)} \left( \mathfrak{k} \right) \otimes H^{(n|0)} \left( \mathfrak{h} \right) \ ,
\end{equation}
since this part of the differential selects the cohomology of superforms in the Lie sub-superalgebra $\mathfrak{h}$ and it does not act on the first factor. 

For the other part of the differential \eqref{TSSL}, we can express its action as
\begin{equation}\label{TSSO}
	\mathcal{Y}^*_{\mathfrak{k}} \mathcal{Y}^*_{\mathfrak{h}} \iota_{\mathcal{Y}_\mathfrak{k}} \left( \omega_{\mathfrak{k}} \otimes \omega_{\mathfrak{h}} \right) = \left( \mathcal{Y}^*_{\mathfrak{k}} \mathcal{Y}^*_{\mathfrak{h}} \iota_{\mathcal{Y}_\mathfrak{k}} \otimes 1 \right) \left( \omega_{\mathfrak{k}} \otimes \omega_{\mathfrak{h}} \right) = 0 \ \iff \ \left( \mathcal{L}_{\mathfrak{h}} \otimes 1 \right) \left( \omega_{\mathfrak{k}} \otimes \omega_{\mathfrak{h}} \right) = 0 \ ,
\end{equation}
on the form $\omega_{\mathfrak{k}} \otimes \omega_{\mathfrak{h}} \in C^{(m|2)} \left( \mathfrak{k} \right)  \otimes C^{(n|0)} \left( \mathfrak{h} \right)$, where with $\mathcal{L}_{\mathfrak{h}}$ we formally denote the Lie derivative along any vector in $\mathfrak{h}$. The double implication follows from the equivalence $d_0 \omega_{\mathfrak{k}} = d \omega_{\mathfrak{k}}$, which follows 
%
since the other two terms of \eqref{TSSK} vanish.  
This means that
\begin{equation}\label{TSSQ}
	H \left( C^{(m|2)} \left( \mathfrak{k} \right) \otimes C^{(n|0)} \left( \mathfrak{h} \right) , \mathcal{Y}^*_{\mathfrak{k}} \mathcal{Y}^*_{\mathfrak{h}} \iota_{\mathcal{Y}_\mathfrak{k}} \right) = \left( C^{(m|2)} \left( \mathfrak{k} \right) \right)^{\mathfrak{h}} \otimes C^{(n|0)} \left( \mathfrak{h} \right) \ ,
\end{equation}
i.e., we have to take $\mathfrak{h}$-\emph{invariant} forms in $C^{(m|2)} \left( \mathfrak{k} \right)$. Finally, we get
\begin{equation}\label{TSSR}
	E_1^{m,n} = \left( C^{(m|2)} \left( \mathfrak{k} \right) \right)^{\mathfrak{h}} \otimes H^{(n|0)} \left( \mathfrak{h} \right) \ .
\end{equation}
The following page of the spectral sequence is defined as
\begin{eqnarray}\label{TSSS}
	E_2^{m,n} \defeq H \left( E_1^{m,n} , d_1 \right) \ , 
\end{eqnarray}
where  $d_1$ is the \emph{vertical operator} $\displaystyle d_1 : E_1^{m,n} \to  E_1^{m+1,n}$, which increases the form number in the coset direction by one. In this case $d_1$ is trivial, as one can readily see from \eqref{TSSK}.

In a general setting, the CE differential reads
\begin{equation}\label{TSST}
	d = \mathcal{Y}^*_{\mathfrak{h}} \mathcal{Y}^*_{\mathfrak{h}} \iota_{\mathcal{Y}_\mathfrak{h}} + \mathcal{Y}^*_{\mathfrak{k}} \mathcal{Y}^*_{\mathfrak{k}} \iota_{\mathcal{Y}_\mathfrak{h}} + \mathcal{Y}^*_{\mathfrak{k}} \mathcal{Y}^*_{\mathfrak{h}} \iota_{\mathcal{Y}_\mathfrak{k}} + \mathcal{Y}^*_{\mathfrak{k}} \mathcal{Y}^*_{\mathfrak{k}} \iota_{\mathcal{Y}_\mathfrak{k}} + \mathcal{Y}^*_{\mathfrak{k}} \mathcal{Y}^*_{\mathfrak{h}} \iota_{\mathcal{Y}_\mathfrak{h}} \ ,
\end{equation}
so that $d_1$ in general reads
\begin{equation}\label{TSSU}
	d_1 = \mathcal{Y}^*_{\mathfrak{k}} \mathcal{Y}^*_{\mathfrak{k}} \iota_{\mathcal{Y}_\mathfrak{k}} + \mathcal{Y}^*_{\mathfrak{k}} \mathcal{Y}^*_{\mathfrak{h}} \iota_{\mathcal{Y}_\mathfrak{h}} \ .
\end{equation}
The first term amounts to the CE differential which is used to calculate the relative cohomology, defined in Def.\ref{Relative Lie (super)algebra cohomology} and in \eqref{VSPL} for the specific case where it is trivial as shown in \eqref{VSPM}; the second term is zero for reductive sub-(super)algebras in the ambient superalgebra, which is true in the case under examination. Then we get
\begin{equation}\label{TSSV}
	\left( C^{(p|2)} \left( \mathfrak{k} \right) \right)^{\mathfrak{h}} = H^{(p|2)} \left( \mathfrak{g} , \mathfrak{h} \right) \ \ \implies \ \ E_2^{m,n} = H^{(m|2)} \left( \mathfrak{g} , \mathfrak{h} \right) \otimes H^{(n|0)} \left( \mathfrak{h} \right) = E_1^{m,n} \ .
\end{equation}
Page 2 for the guiding example is reported in Table \ref{TableTSSC}.
\begin{table}[ht!]
\centering
\begin{tabular}{ccccccccc}
\multicolumn{1}{c|}{$\ldots$} & \multicolumn{1}{c|}{$\ldots$} & \multicolumn{1}{c|}{$\vdots$} & \multicolumn{1}{c|}{$\ldots$}                                                            & \multicolumn{1}{c|}{$\ldots$} & \multicolumn{1}{c|}{$\ldots$} & \multicolumn{1}{c|}{$\ldots$}                                                                                                        & \multicolumn{1}{c|}{$\ldots$} & $\ldots$ \\ \cline{1-2} \cline{4-9} 
\multicolumn{1}{c|}{$\ldots$} & \multicolumn{1}{c|}{0}        & \multicolumn{1}{c|}{2}        & \multicolumn{1}{c|}{0}                                                                   & \multicolumn{1}{c|}{0}        & \multicolumn{1}{c|}{0}        & \multicolumn{1}{c|}{0}                                                                                                               & \multicolumn{1}{c|}{0}        & $\ldots$ \\ \cline{1-2} \cline{4-9} 
\multicolumn{1}{c|}{$\ldots$} & \multicolumn{1}{c|}{0}        & \multicolumn{1}{c|}{1}        & \multicolumn{1}{c|}{$H^{(1|2)} \left( \mathfrak{g} , \mathfrak{h} \right)$} & \multicolumn{1}{c|}{0}        & \multicolumn{1}{c|}{0}        & \multicolumn{1}{c|}{$H^{(1|2)} \left( \mathfrak{g} , \mathfrak{h} \right) \otimes H^{(3|0)} \left( \mathfrak{h} \right)$} & \multicolumn{1}{c|}{0}        & $\ldots$ \\ \cline{1-2} \cline{4-9} 
\multicolumn{1}{c|}{$\ldots$} & \multicolumn{1}{c|}{0}        & \multicolumn{1}{c|}{0}        & \multicolumn{1}{c|}{0}                                                                   & \multicolumn{1}{c|}{0}        & \multicolumn{1}{c|}{0}        & \multicolumn{1}{c|}{0}                                                                                                               & \multicolumn{1}{c|}{0}        & $\ldots$ \\ \cline{1-2} \cline{4-9} 
$\ldots$                      & -1                            &                               & 0                                                                                        & 1                             & 2                             & 3                                                                                                                                    & 4                             & $\ldots$ \\ \cline{1-2} \cline{4-9} 
\multicolumn{1}{c|}{$\ldots$} & \multicolumn{1}{c|}{0}        & \multicolumn{1}{c|}{-1}       & \multicolumn{1}{c|}{0}                                                                   & \multicolumn{1}{c|}{0}        & \multicolumn{1}{c|}{0}        & \multicolumn{1}{c|}{0}                                                                                                               & \multicolumn{1}{c|}{0}        & $\ldots$ \\ \cline{1-2} \cline{4-9} 
\multicolumn{1}{c|}{$\ldots$} & \multicolumn{1}{c|}{$\ldots$} & \multicolumn{1}{c|}{$\vdots$} & \multicolumn{1}{c|}{$\ldots$}                                                            & \multicolumn{1}{c|}{$\ldots$} & \multicolumn{1}{c|}{$\ldots$} & \multicolumn{1}{c|}{$\ldots$}                                                                                                        & \multicolumn{1}{c|}{$\ldots$} & $\ldots$
\end{tabular}
\caption{$E_2^{m,n}$ as obtained in \eqref{TSSV}. The integers $n$ and $m$ are spanned on the horizontal and vertical axes, respectively.}\label{TableTSSC}
\end{table}

We should now proceed with the construction of the higher pages until the spectral sequence converges. This is done by considering the differentials $d_s$:
\begin{equation}\label{TSSW}
	d_s: E_s^{m,n} \to E_s^{m+s,n-s+1} \ ,
\end{equation}
induced by the Koszul differential $d$. In particular, we notice that $d_s$ moves vertically by $s$ and horizontally by $1-s$. In the example under examination, all the higher differentials $d_s$ are trivial:
\begin{equation}\label{TSSX}
	d_s = 0 \ , \ \forall s \geq 2 \ .
\end{equation}
This means that $E_1^{m,n} = E_2^{m,n} = E_3^{m,n} = \ldots = E_\infty^{m,n}$, and the non-trivial cohomology spaces are 
\begin{equation}\label{TSSY}
	H^{1|2} \left( \mathfrak{g} \right) = \Pi \mathbb{K} \ , \ H^{4|2} \left( \mathfrak{g} \right) = \mathbb{K} \ ,
\end{equation}
which are pseudoform cohomology spaces for $\mathfrak{g}$, i.e., with non-zero and non-maximal picture number.

In Section \ref{Pseudoforms as Infinite-Dimensional Representations} we have shown how pseudoforms can be constructed in an algebraic context. In the previous paragraphs, we have shown how to extend the spectral sequence computations of Koszul and Hochschild-Serre to the complex of pseudoforms induced by the sub-(super)algebra $\mathfrak{osp}(1|2)$ in the algebra $\mathfrak{osp}(2|2)$. This is somehow related to the case of sub-supermanifolds with non-trivial odd codimension (see \cite{Witten}).

Second, by considering the case of $\mathfrak{osp} \left( 2|2 \right) $, the superform cohomology (as shown in \cite{Fuks}) is  
 \begin{equation}\label{TSSZ}
	H^{p} \left( \mathfrak{osp} \left( 2|2 \right) \right) \cong H^p \left( \mathfrak{sp} \left( 2 \right) \right) = \begin{cases}
		\Pi^p \mathbb{K} \ , \ \text{if } p = 0,3 \ , \\
		\left\lbrace 0 \right\rbrace \ , \ \text{else.}
	\end{cases} \ ,
\end{equation}
According to \eqref{TSSZ}, the abelian sub-algebra $\mathfrak{so} (2)$ plays no role in $H^p \left( \mathfrak{osp} \left( 2|2 \right) \right)$; what we are going to see is that this sub-algebra (or better, the invariant related to this abelian factor) emerges when considering the other complexes of forms: for that we have to complete the cohomology with pseudoforms and integral forms. 

The pseudoforms we obtained in \eqref{TSSY} are not the only ones for $\mathfrak{osp}(2|2)$. There are additional 
pseudoforms that arise from a new filtration which is inequivalent to \eqref{TSSA}. In \cite{Witten}, superforms and integral forms were introduced as inequivalent modules of Clifford-Weyl algebras; there, it is shown that while for conventional manifolds the irreducible modules are all isomorphic (in other words, there is a single complex of forms related to a manifold), for supermanifolds the modules constructed from a state annihilated by all the contractions (along odd vector fields) and the modules constructed from a state annihilated by form multiplication (by even differential forms) are inequivalent, and they are identified with superforms and integral forms, respectively. In our algebraic setting, this suggests a different way in which we can select a filtration with respect to the one in \eqref{TSSA}, i.e.,
\begin{equation}\label{TSSZA}
	\tilde{F}^p C^{(q|2)} \left( \mathfrak{g} \right) = \left\lbrace \omega \in C^{(q|2)} \left( \mathfrak{g} \right) : \forall \mathcal{Y}^{*a} \in \mathfrak{h}^* , \mathcal{Y}^{*a_1} \wedge \ldots \wedge \mathcal{Y}^{*a_{q+1-p}} \wedge \omega = 0 \right\rbrace \ .
\end{equation}
From \eqref{TSSZA} one can easily verify that
\begin{align}
	\label{TSSZB} \tilde{F}^{q+1} C^{(q|2)} \left( \mathfrak{g} \right) = \tilde{F}^{q+2} C^{(q|2)} \left( \mathfrak{g} \right) = \ldots &= \tilde{F}^{q+n} C^{(q|2)} \left( \mathfrak{g} \right) = 0 \ , \ \forall n \in \mathbb{N} \setminus \left\lbrace 0 \right\rbrace , q \in \mathbb{Z} \ , \\
	\label{TSSZC} \tilde{F}^{p+1} C^{(q|2)} \left( \mathfrak{g} \right) &\subseteq \tilde{F}^{p} C^{(q|2)} \left( \mathfrak{g} \right) \ , \ \forall p,q \in \mathbb{Z} \ ,
\end{align}
analogously to \eqref{TSSB} and \eqref{TSSC}.

Again, we should verify that \eqref{TSSZA} correctly defines a filtration, namely that
\begin{equation}\label{TSSZD}
	d \tilde{F}^p C^{(q|2)} \left( \mathfrak{g} \right) \subseteq \tilde{F}^p C^{(q+1|2)} \left( \mathfrak{g} \right) \ , \ \forall p,q \in \mathbb{Z} \ .
\end{equation}
We can verify this for $p=q$, and by using the same manipulations the general proof follows. We have
\begin{equation}\label{TSSZE}
	\tilde{F}^q C^{(q|2)} \left( \mathfrak{g} \right) = \left\lbrace \omega \in C^{(q|2)} \left( \mathfrak{g} \right) : \forall \mathcal{Y}^{*a} \in  \mathfrak{h}^* , \mathcal{Y}^{*a} \wedge \omega = 0 \right\rbrace \ , $$ $$ \tilde{F}^q C^{(q+1|2)} \left( \mathfrak{g} \right) = \left\lbrace \omega \in C^{(q+1|2)} \left( \mathfrak{g} \right) : \forall \mathcal{Y}^{*a} \in  \mathfrak{h}^* , \mathcal{Y}^{*a_1} \wedge \mathcal{Y}^{*a_2} \wedge \omega = 0 \right\rbrace \ .
\end{equation}
In order to verify that $d \omega \in \tilde{F}^q C^{(q+1|2)} \left( \mathfrak{g} \right)$, we have to check that $\mathcal{Y}^{*a_1} \wedge \mathcal{Y}^{*a_2} \wedge d\omega = 0, \forall \mathcal{Y}^{*a_1} , \mathcal{Y}^{*a_2} \in \mathfrak{h}^*$. In particular, we have
\begin{equation}\label{TSSZF}
	\mathcal{Y}^{*a_1} \wedge \mathcal{Y}^{*a_2} \wedge d\omega = \mathcal{Y}^{*a_1} \wedge \left[ \left( -1 \right)^{\left| \mathcal{Y}^{*a_2} \right|} d \left( \mathcal{Y}^{*a_2} \wedge \omega \right) - \left( -1 \right)^{\left| \mathcal{Y}^{*a_2} \right|} \left( d \mathcal{Y}^{*a_2} \right) \wedge \omega \right] = $$ $$ = \left( -1 \right)^{\left| \mathcal{Y}^{*a_1} \right| + \left| \mathcal{Y}^{*a_2} \right|} \left( d \mathcal{Y}^{*a_2} \right) \wedge \mathcal{Y}^{*a_1} \wedge \omega = 0 \ ,
\end{equation}
where we used the (graded) Leibniz rule $\displaystyle d \left( \mathcal{Y}^{*a_2} \wedge \omega \right) = \left( d \mathcal{Y}^{*a_2} \right) \wedge \omega + \left( -1 \right)^{\left| \mathcal{Y}^{*a_2} \right|} \mathcal{Y}^{*a_2} \wedge \left( d \omega \right)$ and the fact that $\omega \in \tilde{F}^q C^{(q|2)} \left( \mathfrak{g} \right)$. Note that \eqref{TSSZA} correctly defines a filtration on $C^{(\bullet|2)} \left( \mathfrak{g} \right)$.

Analogously to \eqref{TSSH}, we have that the spaces of the filtration \eqref{TSSZA} are explicitly given by
\begin{equation}\label{TSSZG}
	\tilde{F}^{p} C^{(q|2)} \left( \mathfrak{g} \right) \equiv C^{(q|2)}_{* 2q-p-3} \left( \mathfrak{g} \right) \defeq \bigoplus_{i=0}^{q-p} C^{(q-3+i|0)} \left( \mathfrak{k} \right) \otimes C^{(3-i|2)} \left( \mathfrak{h} \right) \ ,
\end{equation}
that is, $\displaystyle C^{(q|2)}_{* 2q-p-3} \left( \mathfrak{g} \right) $ is the (finite-dimensional) sub-space of $\displaystyle C^{(q|2)} \left( \mathfrak{g} \right)$ forms with \emph{at most} $(2q-p-3)$-superforms of 
$\mathfrak{k}^*$. Notice that $C^{(3-i|2)} \left( \mathfrak{h} \right)$ is obtained by acting with $i$ contractions on the Berezinian of the sub-superalgebra $\mathfrak{h}$ (or, in the analogous polyvector realisation, by tensoring $\mathpzc{B}er \left( \mathfrak{h} \right)$ with $i$ parity-changed vectors of $\mathfrak{h}$); in this way we can conveniently interpret $\Omega^{(q|2)}_{\ast 2q-p-3} \left( \mathfrak{g} \right)$ in \eqref{TSSZG} as a space which depends \emph{at most} on $q-p$ contractions along vectors in $\mathfrak{h}$.

The construction of the filtration \eqref{TSSZA} is easily generalisable. We used the fact that $\mathfrak{g} = \mathfrak{osp}(2|2)$ and that $\mathfrak{h}=\mathfrak{osp}(1|2)$ to set the upper bound of $C^{(q|2)} \left( \mathfrak{h} \right)$ to $C^{(3|2)} \left( \mathfrak{h} \right)$ for $q \geq 3$. In the following section we will discuss the general definitions.

From \eqref{TSSZA} we can easily construct page 0 of the spectral sequence. It is more convenient to shift  the spaces in order to display page 0 in an analogous fashion to that of Table \ref{TableTSSB}:
\begin{equation}\label{TSSZI}
	\begin{cases}
		m \to m - 3 +2n \ , \\
		n \to 3-n \ ,
	\end{cases} \ \implies \ \tilde{E}_0^{m,n} \to \tilde{E}_0^{m,n} = \tilde{F}^{m-3+2n} C^{(m+n|2)} \left( \mathfrak{g} \right) / \tilde{F}^{m-2+2n} C^{(m+n|2)} \left( \mathfrak{g} \right) \ .
\end{equation}
Page 0 then reads as in table \ref{TableTSSF}.
\begin{table}[ht!]
\centering 
\begin{tabular}{ccccccc}                            
& \multicolumn{1}{c}0                                                                                                                    
& \multicolumn{1}{c}1                                                                                                                    
& \multicolumn{1}{c}2                                                                                                                    
& \multicolumn{1}{c}3                                                                                                                    
& \multicolumn{1}{c}4                             
& \multicolumn{1}{c}{$\ldots$} 
\\ \cline{2-7} 
 \multicolumn{1}{c|}{$\vdots$} 
& \multicolumn{1}{c|}{$\ldots$}                                                                                        
& \multicolumn{1}{c|}{$\ldots$}                                                                                        
& \multicolumn{1}{c|}{$\ldots$}                                                                                        
& \multicolumn{1}{c|}{$\ldots$}                                                                                       
& \multicolumn{1}{c|}{$\ldots$} 
& $\ldots$ 
\\  \cline{2-7} 
 \multicolumn{1}{c|}{4}        
& \multicolumn{1}{c|}{$C^{(0|2)} \left( \mathfrak{h} \right) \otimes C^{(4|0)} \left( \mathfrak{k} \right)$} 
& \multicolumn{1}{c|}{$C^{(1|2)} \left( \mathfrak{h} \right) \otimes C^{(4|0)} \left( \mathfrak{k} \right)$} 
& \multicolumn{1}{c|}{$C^{(2|2)} \left( \mathfrak{h} \right) \otimes C^{(4|0)} \left( \mathfrak{k} \right)$} 
& \multicolumn{1}{c|}{$C^{(3|2)} \left( \mathfrak{h} \right) \otimes C^{(4|0)} \left( \mathfrak{k} \right)$} 
& \multicolumn{1}{c|}{0}        
& $\ldots$ 
\\ \cline{2-7}  
 \multicolumn{1}{c|}{3}        
& \multicolumn{1}{c|}{$C^{(0|2)} \left( \mathfrak{h} \right) \otimes C^{(3|0)} \left( \mathfrak{k} \right)$} 
& \multicolumn{1}{c|}{$C^{(1|2)} \left( \mathfrak{h} \right) \otimes C^{(3|0)} \left( \mathfrak{k} \right)$} 
& \multicolumn{1}{c|}{$C^{(2|2)} \left( \mathfrak{h} \right) \otimes C^{(3|0)} \left( \mathfrak{k} \right)$} 
& \multicolumn{1}{c|}{$C^{(3|2)} \left( \mathfrak{h} \right) \otimes C^{(3|0)} \left( \mathfrak{k} \right)$} 
& \multicolumn{1}{c|}{0}        
& $\ldots$ 
\\ \cline{2-7}
 \multicolumn{1}{c|}{2}        
& \multicolumn{1}{c|}{$C^{(0|2)} \left( \mathfrak{h} \right) \otimes C^{(2|0)} \left( \mathfrak{k} \right)$} 
& \multicolumn{1}{c|}{$C^{(1|2)} \left( \mathfrak{h} \right) \otimes C^{(2|0)} \left( \mathfrak{k} \right)$} 
& \multicolumn{1}{c|}{$C^{(2|2)} \left( \mathfrak{h} \right) \otimes C^{(2|0)} \left( \mathfrak{k} \right)$} 
& \multicolumn{1}{c|}{$C^{(3|2)} \left( \mathfrak{h} \right) \otimes C^{(2|0)} \left( \mathfrak{k} \right)$} 
& \multicolumn{1}{c|}{0}        
& $\ldots$ 
\\ \cline{2-7} 
\multicolumn{1}{c|}{1}        
& \multicolumn{1}{c|}{$C^{(0|2)} \left( \mathfrak{h} \right) \otimes C^{(1|0)} \left( \mathfrak{k} \right)$} 
& \multicolumn{1}{c|}{$C^{(1|2)} \left( \mathfrak{h} \right) \otimes C^{(1|0)} \left( \mathfrak{k} \right)$} 
& \multicolumn{1}{c|}{$C^{(2|2)} \left( \mathfrak{h} \right) \otimes C^{(1|0)} \left( \mathfrak{k} \right)$} 
& \multicolumn{1}{c|}{$C^{(3|2)} \left( \mathfrak{h} \right) \otimes C^{(1|0)} \left( \mathfrak{k} \right)$} 
& \multicolumn{1}{c|}{0}        
& $\ldots$ 
\\ \cline{2-7} 
 \multicolumn{1}{c|}{0}        
& \multicolumn{1}{c|}{$C^{(0|2)} \left( \mathfrak{h} \right)$}                                                    
& \multicolumn{1}{c|}{$C^{(1|2)} \left( \mathfrak{h} \right)$}                                                    
& \multicolumn{1}{c|}{$C^{(2|2)} \left( \mathfrak{h} \right)$}                                                    
& \multicolumn{1}{c|}{$C^{(3|2)} \left( \mathfrak{h} \right)$}                                                    
& \multicolumn{1}{c|}{0}        
& $\ldots$ 
\\ \cline{2-7}  
 \multicolumn{1}{c|}{-1}       
& \multicolumn{1}{c|}{0}                                                                                               
& \multicolumn{1}{c|}{0}                                                                                               
& \multicolumn{1}{c|}{0}                                                                                               
& \multicolumn{1}{c|}{0}                                                                                               
& \multicolumn{1}{c|}{0}        
& $\ldots$ 
\\ \cline{2-7} 
 \multicolumn{1}{c|}{$\vdots$} 
& \multicolumn{1}{c|}{$\ldots$}                                                                                        
& \multicolumn{1}{c|}{$\ldots$}                                                                                        
& \multicolumn{1}{c|}{$\ldots$}                                                                                        
& \multicolumn{1}{c|}{$\ldots$}                                                                                        
& \multicolumn{1}{c|}{$\ldots$} 
& $\ldots$
\end{tabular}
\caption{$\tilde{E}_0^{m,n}$ as defined in \eqref{TSSZI}. The integers $n$ and $m$ are spanned on the horizontal and vertical axes, respectively.}\label{TableTSSF}
\end{table}

Notice that, by contrast with table \ref{TableTSSB}, table \ref{TableTSSF} has non-trivial entries in the first and second quadrants, as a consequence of the facts that integral forms of $\mathfrak{h}$ are unbounded from below and that superforms of $\mathfrak{k}$ are bounded from below. The differentials that we use to define the next pages of the spectral sequence are the same as before, therefore we can construct page 1 as
\begin{equation}\label{TSSZJ}
	d_0 = \mathcal{Y}^{*}_{\mathfrak{k}} \mathcal{Y}^{*}_{\mathfrak{h}} \iota_{\mathcal{Y}_{\mathfrak{k}}} + \mathcal{Y}^{*}_{\mathfrak{h}} \mathcal{Y}^{*}_{\mathfrak{h}} \iota_{\mathcal{Y}_\mathfrak{h}} \ , \ \implies \ \tilde{E}_1^{m,n} \defeq H \left( \tilde{E}_0^{m,n} , d_0 \right) \ .
\end{equation}
In exact analogy to the previous case, we have
\begin{equation}\label{TSSZK}
	H \left( C^{(p|2)} \left( \mathfrak{h} \right) \otimes C^{(q|0)} \left( \mathfrak{k} \right) , d_0 \right) = H^{(p|2)} \left( \mathfrak{h} \right) \otimes \left( C^{(q|0)} \left( \mathfrak{k} \right) \right)^{\mathfrak{h}} \ .
\end{equation}
In the specific case under examination, we have
\begin{equation}\label{TSSZL}
	H^{(p|2)} \left( \mathfrak{h} \right) = \begin{cases}
		\mathbb{K} \ , \ \text{if} \ p = 0 \ ,\\
		\Pi \mathbb{K} \ , \ \text{if} \ p = 3 \ ,\\
		\left\lbrace 0 \right\rbrace \ , \ \text{else} \ ,
	\end{cases} \ , \ \left( C^{(q|0)} \left( \mathfrak{k} \right) \right)^{\mathfrak{h}} = \begin{cases}
		\mathbb{K} \ , \ \text{if} \ q = 0 \ ,\\
		\left\lbrace 0 \right\rbrace \ , \ \text{else} \ .
	\end{cases}
\end{equation}
In particular, the cohomology spaces $H^{(p|2)} \left( \mathfrak{h} \right)$ can be easily determined starting from the spaces $H^{(p|0)} \left( \mathfrak{h} \right)$, via the map \eqref{IFLSN}:
\begin{equation}\label{TSSZLA}
	H^{(p|0)} \left( \mathfrak{h} \right) = \begin{cases}
		\mathbb{K} \ , \ \text{if} \ p = 0 \ ,\\
		\Pi \mathbb{K} \ , \ \text{if} \ p = 3 \ ,\\
		\left\lbrace 0 \right\rbrace \ , \ \text{else} \ ,
	\end{cases} \ \overset{\star}{\underset{}{\longleftrightarrow}} \ \ H^{(3-p|2)} \left( \mathfrak{h} \right) = \begin{cases}
		\Pi \mathbb{K} \ , \ \text{if} \ p = 0 \ ,\\
		\mathbb{K} \ , \ \text{if} \ p = 3 \ ,\\
		\left\lbrace 0 \right\rbrace \ , \ \text{else} \ ,
	\end{cases} \ .
\end{equation}

As we discussed above, this already gives page 2 of the spectral sequence, since the differential $d_1$ is trivial:
\begin{equation}\label{TSSZLAA}
	\tilde{E}^{m,n}_2 \defeq H \left( \tilde{E}_1^{m,n} , d_1 \right) = \tilde{E}_1^{m,n} \ .
\end{equation}
In Table \ref{TableTSSG} we write page $\tilde{E}_2^{m,n}$:
\begin{table}[ht!]
\centering
\begin{tabular}{cccccccc}
\multicolumn{1}{c|}{$\ldots$} & \multicolumn{1}{c|}{$\vdots$} & \multicolumn{1}{c|}{$\ldots$}                                & \multicolumn{1}{c|}{$\ldots$} & \multicolumn{1}{c|}{$\ldots$} & \multicolumn{1}{c|}{$\ldots$}                                & \multicolumn{1}{c|}{$\ldots$} & $\ldots$ \\ \cline{1-1} \cline{3-8} 
\multicolumn{1}{c|}{$\ldots$} & \multicolumn{1}{c|}{2}        & \multicolumn{1}{c|}{0}                                       & \multicolumn{1}{c|}{0}        & \multicolumn{1}{c|}{0}        & \multicolumn{1}{c|}{0}                                       & \multicolumn{1}{c|}{0}        & $\ldots$ \\ \cline{1-1} \cline{3-8} 
\multicolumn{1}{c|}{$\ldots$} & \multicolumn{1}{c|}{1}        & \multicolumn{1}{c|}{0}                                       & \multicolumn{1}{c|}{0}        & \multicolumn{1}{c|}{0}        & \multicolumn{1}{c|}{0}                                       & \multicolumn{1}{c|}{0}        & $\ldots$ \\ \cline{1-1} \cline{3-8} 
\multicolumn{1}{c|}{$\ldots$} & \multicolumn{1}{c|}{0}        & \multicolumn{1}{c|}{$H^{(0|2)} \left( \mathfrak{h} \right)$} & \multicolumn{1}{c|}{0}        & \multicolumn{1}{c|}{0}        & \multicolumn{1}{c|}{$H^{(3|2)} \left( \mathfrak{h} \right)$} & \multicolumn{1}{c|}{0}        & $\ldots$ \\ \cline{1-1} \cline{3-8} 
$\ldots$                      &                               & 0                                                            & 1                             & 2                             & 3                                                            & 4                             & $\ldots$ \\ \cline{1-1} \cline{3-8} 
\multicolumn{1}{c|}{$\ldots$} & \multicolumn{1}{c|}{-1}       & \multicolumn{1}{c|}{0}                                       & \multicolumn{1}{c|}{0}        & \multicolumn{1}{c|}{0}        & \multicolumn{1}{c|}{0}                                       & \multicolumn{1}{c|}{0}        & $\ldots$ \\ \cline{1-1} \cline{3-8} 
\multicolumn{1}{c|}{$\ldots$} & \multicolumn{1}{c|}{$\vdots$} & \multicolumn{1}{c|}{$\ldots$}                                & \multicolumn{1}{c|}{$\ldots$} & \multicolumn{1}{c|}{$\ldots$} & \multicolumn{1}{c|}{$\ldots$}                                & \multicolumn{1}{c|}{$\ldots$} & $\ldots$
\end{tabular}
\caption{$\tilde{E}_2^{m,n}$ as defined in \eqref{TSSZLAA}. The integers $n$ and $m$ are spanned on the horizontal and vertical axes, respectively.}\label{TableTSSG}
\end{table}

The other differentials all vanish as in the previous case, hence the spectral sequence converges at page $\tilde{E}_1^{m,n} = \tilde{E}_2^{m,n} = \ldots = \tilde{E}_\infty^{m,n}$, and in particular
\begin{equation}\label{TSSZM}
	H^{(0|2)} \left( \mathfrak{g} \right) = \mathbb{K} \ , \ H^{(3|2)} \left( \mathfrak{g} \right) = \Pi \mathbb{K} \ .
\end{equation}

The result can be read in a two-fold way: by using the filtration introduced in \eqref{TSSZA}, we have confirmed that there are non-trivial cohomology classes among pseudoforms but also shown that the classes we found are inequivalent to those found in \eqref{TSSY}. This is consistent with the inequivalence of Clifford-Weyl modules defined in \cite{Witten}, as mentioned before.

We can put all the results together in the following proposition:

\begin{prop}\label{osp2|2}
	The CE cohomology of the Lie (super)algebra $\mathfrak{osp} \left( 2|2 \right)$ in the sectors of superforms, integral forms and picture-2 pseudoforms as induced by the sub-(super)algebra $\mathfrak{osp}(1|2)$ reads
	\begin{align}
		\label{TSSZN} H^p \left( \mathfrak{osp} \left( 2|2 \right) \right) &= \begin{cases}
			\Pi^p \mathbb{K} \ , \ \text{if} \ p=0,3 \ , \\
			\left\lbrace 0 \right\rbrace \ , \ \text{else} \ ,
		\end{cases} \\
		\label{TSSZO} H^p_{pseudo,2} \left( \mathfrak{osp} \left( 2|2 \right) \right) &= \begin{cases}
			\Pi^p \mathbb{K} \ , \ \text{if} \ p=0,1,3,4 \ , \\
			\left\lbrace 0 \right\rbrace \ , \ \text{else} \ ,
		\end{cases} \\
		\label{TSSZP} H^p_{int} \left( \mathfrak{osp} \left( 2|2 \right) \right) &= \begin{cases}
			\Pi^p \mathbb{K} \ , \ \text{if} \ p=1,4 \ , \\
			\left\lbrace 0 \right\rbrace \ , \ \text{else} \ ,
		\end{cases}
	\end{align}
	where we indicated $\Pi^p = \Pi$ for $p=1 \mod 2$, $\Pi^p = \text{id}$ for $p=0 \mod 2$ and the subscript \virgolette pseudo,2" is used to emphasise that the pseudoforms are obtained at picture number 2.
\end{prop}

	From \eqref{TSSZN}, \eqref{TSSZO} and \eqref{TSSZP} we see that the duality between superforms and integral forms given by the \virgolette $\star$" map extends (at least for the given example) to the complex of pseudoforms as well. This is a consequence of the fact that we introduced two inequivalent filtrations and calculated the cohomology in the two cases. We conjecture that this result should hold for any basic Lie superalgebra. This would represent the extension to Lie superalgebras of the Poincar\'e duality: for ordinary Lie algebras, one has
	\begin{equation}
		H^p \left( \mathfrak{g} \right) \cong H^{m-p} \left( \mathfrak{g} \right) \ ,
	\end{equation}
where the top form $\omega \in H^{m} \left( \mathfrak{g} \right)$ is the fulcrum of the duality. In \cite{CCGN2} the authors used the same argument to demonstrate the isomorphism between superform and integral form cohomologies. This shows that, if Poincar\'e duality can be extended to superalgebras, it may involve forms in different complexes, i.e., with different picture number. We then conjecture the existence of an isomorphism
	\begin{equation}
		\star : H^{(\bullet|\bullet)} \left( \mathfrak{g} , \mathbb{K} \right) \overset{\cong}{\underset{}{\longrightarrow}} H^{(m-\bullet|n-\bullet)} \left( \mathfrak{g} , \mathbb{K} \right) \ ,
	\end{equation}
	which would extend the Poincar\'e duality to superalgebras, by considering any form complex. In the following section we will introduce the definition of a generalised page zero that implements both the pseudoforms emerging from $\mathfrak{k}$ and the pseudoforms emerging from $\mathfrak{h}$.

\section{General Constructions}

In the previous sections, we have shown how to introduce two inequivalent filtrations that we have used to calculate the pseudoform cohomology for the specific example $\mathfrak{osp}(2|2)$. We have shown that \emph{both} filtrations should be considered to reproduce all the possible pseudoforms induced by the choice of a sub-(super)algebra. As we emphasised many times, the construction of such filtrations is very general and can be extended to any superalgebra $\mathfrak{g}$ with a sub-superalgebra $\mathfrak{h}$. Here we deal with the general constructions.

We start by generalising the definitions \eqref{TSSA} and \eqref{TSSZA} that keep into account both the sectors that define page 0 of the spectral sequence. The other definitions will follow directly from Def.\ref{Koszul-Hochschild-Serre Spectral Sequence}; the existence, the convergence and all the usual properties of spectral sequences are guaranteed by Def.\ref{Pseudoform Cohomology Relative to a sub-(super)algebra}: dealing with pseudoform complexes simply amount to deal with the $\mathfrak{g}$-modules $V_\bullet^{(\bullet|\bullet)} \otimes V$ instead of $V$.

\begin{defn}[Generalised KHS Spectral Sequence]\label{Generalised KHS Spectral Sequence}

	Given a Lie (super)algebra $\mathfrak{g}$ and a Lie sub-(super)algebra $\mathfrak{h}$ a $\mathfrak{g}$-module $V$\footnote{In the following, we will refer to the trivial module $V = \mathbb{K}$ only.} and denoting $\mathfrak{k}=\mathfrak{g}/\mathfrak{h}$, we define the two inequivalent filtrations
	\begin{eqnarray}
		\label{GCA} F^p C^{(q|l)} \left( \mathfrak{g}, V \right) &\defeq& \left\lbrace \omega \in C^{(q|l)} \left( \mathfrak{g} , V \right) : \forall \mathcal{Y}_a \in \mathfrak{h} , \iota_{\mathcal{Y}_{a_1}} \ldots \iota_{\mathcal{Y}_{a_{q+1-p}}} \omega = 0 \right\rbrace \ , \\
		\label{GCAA} \tilde{F}^{p} C^{(q|l)} \left( \mathfrak{g} , V \right) &\defeq& \left\lbrace \omega \in C^{(q|l)} \left( \mathfrak{g} , V \right) : \forall \mathcal{Y}^{*a} \in \mathfrak{h}^* , \mathcal{Y}^{*a_{1}} \wedge \ldots \wedge \mathcal{Y}^{*a_{q+1-p}} \wedge \omega = 0 \right\rbrace \ ,
	\end{eqnarray}
	where $p,q \in \mathbb{Z}$ and $l \in \left\lbrace 0,\dim \mathfrak{h}_1, \dim \mathfrak{k}_1, \dim \mathfrak{g}_1 \right\rbrace$. For each $l$, associated to \eqref{GCA} and \eqref{GCAA}, there exist a spectral sequence $\displaystyle \left( \mathcal{E}_s^{\bullet,\bullet} , d_s \right)_{s \in \mathbb{N} \cup \left\lbrace 0 \right\rbrace}$ that converges to $H^{(\bullet|l)} \left( \mathfrak{g} , V \right)$. In particular, we can define the \emph{total page zero} of the spectral sequence, for each $l$, as
	\begin{equation}\label{GCB}
		\mathcal{E}_0^{m,n} \defeq E_0^{m,n} \oplus \tilde{E}_0^{m,n} \defeq \frac{F^m C^{(m+n|l)} \left( \mathfrak{g} , V \right)}{F^{m+1} C^{(m+n|l)} \left( \mathfrak{g} , V \right)} \oplus \frac{\tilde{F}^{m+2n-r} C^{(m+n|l)} \left( \mathfrak{g} , V \right)}{\tilde{F}^{m+2n-r+1} C^{(m+n|l)} \left( \mathfrak{g} , V \right)} \ .
	\end{equation}
	The differentials are induced by the Chevalley-Eilenberg differential \eqref{PD} (adapted to act on pseudoforms as in Def.\ref{Pseudoform Cohomology Relative to a sub-(super)algebra}) and, for each picture number $l$, each page of the spectral sequence is defined as the cohomology of the previous one:
	\begin{equation}\label{GCBA}
		d_s : \mathcal{E}^{p,q}_s \to \mathcal{E}^{p+s,q+1-s}_s \ , \ \mathcal{E}^{\bullet,\bullet}_{s+1} \defeq \left( \mathcal{E}^{\bullet,\bullet}_s , d_s \right) \ , \ \mathcal{E}^{\bullet,\bullet}_{\infty} \cong H^{(\bullet | l)} \left( \mathfrak{g} , V \right) \ .
	\end{equation}
	As usual, if $\exists n \in \mathbb{N}\cup \left\lbrace 0 \right\rbrace : d_s = 0 , \forall s \geq n$, then we say that the spectral sequence \emph{converges at page n} and we denote
	\begin{equation}\label{GCBB}
		\mathcal{E}_n^{\bullet,\bullet} = \mathcal{E}_{n+1}^{\bullet,\bullet} = \ldots = \mathcal{E}_\infty^{\bullet,\bullet} \ .
	\end{equation}
\end{defn}

\begin{remark}
	Notice that the procedure actually refers also to the extremal cases of superforms and integral forms. In particular, Def.\ref{Generalised KHS Spectral Sequence} simplifies for $l=0, \dim \mathfrak{g}_1$. In these cases, indeed, one directly sees that either the filtration $F^pC^{(q|l)} \left( \mathfrak{g} \right)$ or the filtration $\tilde{F}^p C^{(q|l)} \left( \mathfrak{g} \right)$ is empty (we are assuming that the sub-superalgebra has non-zero odd dimension):
	\begin{equation}\label{GCC}
		l = 0 \ \implies \ \tilde{F}^p C^{(q|0)} \left( \mathfrak{g} , V \right) = \left\lbrace \omega \in C^{(q|0)} \left( \mathfrak{g} , V \right) : \forall \mathcal{Y}^{*a} \in \mathfrak{h}^* , \mathcal{Y}^{*a_{1}} \wedge \ldots \wedge \mathcal{Y}^{*a_{q+1-p}} \wedge \omega = 0 \right\rbrace = \left\lbrace 0 \right\rbrace \ ,
	\end{equation}
	as a consequence of the fact that there is no top superform;
	\begin{equation}\label{GCD}
		l = \dim \mathfrak{g}_1 \ \implies \ F^p C^{(q|\dim \mathfrak{g}_1)} \left( \mathfrak{g} , V \right) = \left\lbrace \omega \in C^{(q|\dim \mathfrak{g}_1)} \left( \mathfrak{g} , V \right) : \forall \mathcal{Y}_a \in \mathfrak{h} , \iota_{\mathcal{Y}_{a_1}} \ldots \iota_{\mathcal{Y}_{a_{q+1-p}}} \omega = 0 \right\rbrace = \left\lbrace 0 \right\rbrace \ ,
	\end{equation}
	as a consequence of the fact that there is no bottom integral form. In these two cases, page zero in \eqref{GCB} then simplifies as:
	\begin{align}
		\label{GCE} l = 0 \ \implies \ & \mathcal{E}_0^{m,n} = E_0^{m,n} = \frac{F^m C^{(m+n|0)} \left( \mathfrak{g} , V \right)}{F^{m+1} C^{(m+n|0)} \left( \mathfrak{g} , V \right)} \ , \\
		\label{GCF} l = \dim \mathfrak{g}_1 \ \implies \ & \mathcal{E}_0^{m,n} = \tilde{E}_0^{m,n} = \frac{\tilde{F}^{m+2n-r} C^{(m+n|\dim \mathfrak{g}_1)} \left( \mathfrak{g} , V \right)}{\tilde{F}^{m+2n-r+1} \Omega^{(m+n|\dim \mathfrak{g}_1)} \left( \mathfrak{g} , V \right)} \ .
	\end{align}
	This in particular shows that integral forms are kept into account when using spectral sequences via the second filtration \eqref{GCAA}.
	
	Analogous simplifications occur when dealing with pseudoforms complexes: if $\dim \mathfrak{h}_1 \neq \dim \mathfrak{k}_1$ we have
	\begin{equation}\label{GCFA}
		l = \dim \mathfrak{h}_1 \ \implies \ F^p C^{(q|\dim \mathfrak{h}_1)} \left( \mathfrak{g} , V \right) = \left\lbrace 0 \right\rbrace \ ,
	\end{equation}
	as in this case we are considering pseudoforms built out of integral forms in $\mathfrak{h}$, hence there is no lower bound when acting with contraction along vectors in $\mathfrak{h}$. On the other hand, we have
	\begin{equation}\label{GCFB}
		l = \dim \mathfrak{k}_1 \ \implies \ \tilde{F}^p C^{(q|\dim \mathfrak{k}_1)} \left( \mathfrak{g} , V \right) = \left\lbrace 0 \right\rbrace \ ,
	\end{equation}
	since now we are considering pseudoforms built out of superforms in $\mathfrak{h}^*$, hence there is no upper bound when multiplying by forms in $\mathfrak{h}^*$. On the other hand, if $\dim \mathfrak{h}_1 = \dim \mathfrak{k}_1$ it follows that both filtrations are non-empty, as we have shown for the example of $\mathfrak{g} = \mathfrak{osp}(2|2)$ and $\mathfrak{h} = \mathfrak{osp}(1|2)$.
\end{remark}

\begin{remark}
In \eqref{TSSZA} and in the paragraphs above, we emphasised that the complementary filtration related to the modules obtained by contraction operators and multiplication by forms are inequivalent. In particular, this is true for sub-superalgebras with non-trivial odd dimensions. On the other hand, if one introduces a purely even sub-algebra $\mathfrak{h}$, the filtrations $F^p C^{(q|l)} \left( \mathfrak{g} \right)$ and $\tilde{F}^p C^{(q|l)} \left( \mathfrak{g} \right)$ induce
	\begin{equation}\label{GCG}
		E_0^{m,n} \cong \tilde{E}_0^{m,n} \ ,
	\end{equation}
	 and $\mathcal{E}^{m,n}_0$ counts the spaces twice. The equivalence follows directly from the definitions \eqref{GCA} and \eqref{GCAA}, for any picture number $l$ which in this case can only assume the values $l=0, \dim \mathfrak{g}_1$:
	 \begin{align}
	 	\label{GCGA} F^m C^{(m+n|l)} \left( \mathfrak{g} , V \right) = \bigoplus_{i=0}^n C^{i} \left( \mathfrak{h} \right) \otimes C^{(m+n-i|l)} \left( \mathfrak{k} \right) \otimes V \ &\implies \ E_0^{m,n} = C^{n} \left( \mathfrak{h} \right) \otimes C^{(m|l)} \left( \mathfrak{k} \right) \otimes V \ , \\
	 	\label{GCGB} \tilde{F}^{m+2n-r} C^{(m+n|l)} \left( \mathfrak{g} , V \right) = \bigoplus_{i=n}^{\text{dim} \left( \mathfrak{h} \right)} C^{i} \left( \mathfrak{h} \right) \otimes C^{(m+n-i|l)} \left( \mathfrak{k} \right) \otimes V \ &\implies \ \tilde{E}_0^{m,n} = C^{n} \left( \mathfrak{h} \right) \otimes C^{(m|l)} \left( \mathfrak{k} \right) \otimes V \ .
	 \end{align}
	
	It is interesting to note that this exactly corresponds to what is done for the calculation of the algebraic superform cohomology for Lie superalgebras as in \cite{Fuks}. There, one does not need to introduce the second filtration. The drawback is that since $\mathfrak{h}$ does not have fermionic generators, it does not induce any pseudoform module. In order to calculate pseudoform cohomology classes at a given picture number, one has to calculate them by using an explicit realisation. In the following section, we show how to use the distributional realisation for pseudoforms and integral forms to calculate the cohomology spaces explicitly in the case of a bosonic sub-algebra for the example $\mathfrak{osp}(2|2)$.
	
	In \cite{Fuks}, the author uses $\mathfrak{h} = \mathfrak{g}_0$, and since the super-coset is purely odd, the commutation relations can be formally written as
	\begin{equation}\label{GCH}
		\left[ \mathfrak{h} , \mathfrak{h} \right] \subseteq \mathfrak{h} \ , \ \left[ \mathfrak{k} , \mathfrak{k} \right] \subseteq \mathfrak{h} \ , \ \left[ \mathfrak{h} , \mathfrak{k} \right] \subseteq \mathfrak{k} \ ,
	\end{equation}
	indicating that $\mathfrak{h}$ is reductive and $\mathfrak{k}$ is homogeneous. For the sake of clarity, we quickly review the superform cohomology (with values in the trivial module) in this case, for the example of the previous section $\mathfrak{osp}(2|2)$. We have that page zero is given by
	\begin{equation}\label{GCHA}
		E_0^{m,n} = C^{n} \left( \mathfrak{h} \right) \otimes C^{(m|0)} \left( \mathfrak{k} \right) \ .
	\end{equation}
	It is not difficult to prove (see Sec.\ref{Cohomology via Spectral Sequences}) that, because of the structure \eqref{GCH}, page 2 of the spectral sequence at picture number 0 is given by
	\begin{equation}\label{GCHB}
		E_2^{m,n} = H^{n} \left( \mathfrak{h} \right) \otimes H^{(m|0)} \left( \mathfrak{g} , \mathfrak{h} \right) \ .
	\end{equation}
	In particular, one has
	\begin{equation}\label{GCHC}
		H^{n} \left( \mathfrak{h} \right) = H^{n} \left( \mathfrak{so}(2) \oplus \mathfrak{sp}(2) \right) = \begin{cases}
			\mathbb{K} \ , \ \text{if} \ n = 0,4 \ , \\
			\Pi \mathbb{K} \ , \ \text{if} \ n = 1,3 \ , \\
			\left\lbrace 0 \right\rbrace \ , \ \text{else}.
		\end{cases}
	\end{equation}
	On the other hand, one has that the relative cohomology of the coset $\mathfrak{k}$ is infinitely generated as (in the following section we will give explicit expressions for the generators)
	\begin{equation}\label{GCHD}
		H^{(m|0)} \left( \mathfrak{g} , \mathfrak{h} \right) = \begin{cases}
			\mathbb{R} \ , \ \text{if} \ m = 0,2,4,\ldots \ , \\
			\left\lbrace 0 \right\rbrace \ , \ \text{if} \ \text{else} \ .
		\end{cases}
	\end{equation}
	We report \eqref{GCHB} in Table \eqref{TableGCsuper}.
	\begin{table}[]
\centering
\resizebox{1.07\textwidth}{!}{\hspace{-1.4cm}
\begin{tabular}{cccccccccc}
\multicolumn{1}{c|}{$\ldots$} & \multicolumn{1}{c|}{$\ldots$} & \multicolumn{1}{c|}{$\vdots$} & \multicolumn{1}{c|}{$\ldots$}                                                                              & \multicolumn{1}{c|}{$\ldots$}                                                                              & \multicolumn{1}{c|}{$\ldots$} & \multicolumn{1}{c|}{$\ldots$}                                                                              & \multicolumn{1}{c|}{$\ldots$}                                                                              & \multicolumn{1}{c|}{$\ldots$} & $\ldots$ \\ \cline{1-2} \cline{4-10} 
\multicolumn{1}{c|}{$\ldots$} & \multicolumn{1}{c|}{0}        & \multicolumn{1}{c|}{4}        & \multicolumn{1}{c|}{$H^{(0|0)} \left( \mathfrak{h} \right) \otimes H^{(4|0)} \left( \mathfrak{k} \right)$} & \multicolumn{1}{c|}{$H^{(1|0)} \left( \mathfrak{h} \right) \otimes H^{(4|0)} \left( \mathfrak{k} \right)$} & \multicolumn{1}{c|}{0}        & \multicolumn{1}{c|}{$H^{(3|0)} \left( \mathfrak{h} \right) \otimes H^{(4|0)} \left( \mathfrak{k} \right)$} & \multicolumn{1}{c|}{$H^{(4|0)} \left( \mathfrak{h} \right) \otimes H^{(4|0)} \left( \mathfrak{k} \right)$} & \multicolumn{1}{c|}{0}        & $\ldots$ \\ \cline{1-2} \cline{4-10} 
\multicolumn{1}{c|}{$\ldots$} & \multicolumn{1}{c|}{0}        & \multicolumn{1}{c|}{3}        & \multicolumn{1}{c|}{0}                                                                                     & \multicolumn{1}{c|}{0}                                                                                     & \multicolumn{1}{c|}{0}        & \multicolumn{1}{c|}{0}                                                                                     & \multicolumn{1}{c|}{0}                                                                                     & \multicolumn{1}{c|}{0}        & $\ldots$ \\ \cline{1-2} \cline{4-10} 
\multicolumn{1}{c|}{$\ldots$} & \multicolumn{1}{c|}{0}        & \multicolumn{1}{c|}{2}        & \multicolumn{1}{c|}{$H^{(0|0)} \left( \mathfrak{h} \right) \otimes H^{(2|0)} \left( \mathfrak{k} \right)$} & \multicolumn{1}{c|}{$H^{(1|0)} \left( \mathfrak{h} \right) \otimes H^{(2|0)} \left( \mathfrak{k} \right)$} & \multicolumn{1}{c|}{0}        & \multicolumn{1}{c|}{$H^{(3|0)} \left( \mathfrak{h} \right) \otimes H^{(2|0)} \left( \mathfrak{k} \right)$} & \multicolumn{1}{c|}{$H^{(4|0)} \left( \mathfrak{h} \right) \otimes H^{(2|0)} \left( \mathfrak{k} \right)$} & \multicolumn{1}{c|}{0}        & $\ldots$ \\ \cline{1-2} \cline{4-10} 
\multicolumn{1}{c|}{$\ldots$} & \multicolumn{1}{c|}{0}        & \multicolumn{1}{c|}{1}        & \multicolumn{1}{c|}{0}                                                                                     & \multicolumn{1}{c|}{0}                                                                                     & \multicolumn{1}{c|}{0}        & \multicolumn{1}{c|}{0}                                                                                     & \multicolumn{1}{c|}{0}                                                                                     & \multicolumn{1}{c|}{0}        & $\ldots$ \\ \cline{1-2} \cline{4-10} 
\multicolumn{1}{c|}{$\ldots$} & \multicolumn{1}{c|}{0}        & \multicolumn{1}{c|}{0}        & \multicolumn{1}{c|}{$H^{(0|0)} \left( \mathfrak{h} \right) \otimes H^{(0|0)} \left( \mathfrak{k} \right)$} & \multicolumn{1}{c|}{$H^{(1|0)} \left( \mathfrak{h} \right) \otimes H^{(0|0)} \left( \mathfrak{k} \right)$} & \multicolumn{1}{c|}{0}        & \multicolumn{1}{c|}{$H^{(3|0)} \left( \mathfrak{h} \right) \otimes H^{(0|0)} \left( \mathfrak{k} \right)$} & \multicolumn{1}{c|}{$H^{(4|0)} \left( \mathfrak{h} \right) \otimes H^{(0|0)} \left( \mathfrak{k} \right)$} & \multicolumn{1}{c|}{0}        & $\ldots$ \\ \cline{1-2} \cline{4-10} 
$\ldots$                      & -1                            &                               & 0                                                                                                          & 1                                                                                                          & 2                             & 3                                                                                                          & 4                                                                                                          & 5                             & $\ldots$ \\ \cline{1-2} \cline{4-10} 
\multicolumn{1}{c|}{$\ldots$} & \multicolumn{1}{c|}{0}        & \multicolumn{1}{c|}{-1}       & \multicolumn{1}{c|}{0}                                                                                     & \multicolumn{1}{c|}{0}                                                                                     & \multicolumn{1}{c|}{0}        & \multicolumn{1}{c|}{0}                                                                                     & \multicolumn{1}{c|}{0}                                                                                     & \multicolumn{1}{c|}{0}        & $\ldots$ \\ \cline{1-2} \cline{4-10} 
\multicolumn{1}{c|}{$\ldots$} & \multicolumn{1}{c|}{$\ldots$} & \multicolumn{1}{c|}{$\vdots$} & \multicolumn{1}{c|}{$\ldots$}                                                                              & \multicolumn{1}{c|}{$\ldots$}                                                                              & \multicolumn{1}{c|}{$\ldots$} & \multicolumn{1}{c|}{$\ldots$}                                                                              & \multicolumn{1}{c|}{$\ldots$}                                                                              & \multicolumn{1}{c|}{$\ldots$} & $\ldots$
\end{tabular}
}
\caption{$E_2^{m,n}$ as defined in \eqref{GCHB}. The integers $n$ and $m$ are spanned on the horizontal and vertical axes, respectively.}\label{TableGCsuper}
\end{table}
	The differential $d_2$ formally reads as
	\begin{equation}\label{GCHE}
		d_2 = V_{\mathfrak{k}} V_{\mathfrak{k}} \iota_\mathfrak{h} \ ,
	\end{equation}
	hence moving in Table \eqref{TableGCsuper} vertically by two and horizontally on the left by one. Page three of the spectral sequence is defined as
	\begin{equation}\label{GCHF}
		E_3^{m,n} \defeq H \left( E_2^{m,n} , d_2 \right) \ ,
	\end{equation}
	and it coincides with the convergence of the spectral sequence since all the higher differentials are trivial. One directly verifies that
	\begin{equation}
		H^{(p|0)} \left( \mathfrak{osp}(2|2) \right) \equiv H^p \left( \mathfrak{osp}(2|2) \right) = H^p \left( \mathfrak{sp}(2) \right) = \begin{cases}
			\mathbb{K} \ , \ \text{if} \ p=0 \ , \\
			\Pi \mathbb{K} \ , \ \text{if} \ p=3 \ , \\
			\left\lbrace 0 \right\rbrace \ , \ \text{else},
		\end{cases}
	\end{equation}
	as demonstrated in \cite{Fuks}.
\end{remark}

We are now ready to extend Thm.\ref{KHStheorem} to the case of superalgebras with the two filtrations and on pseudoforms with picture number $l$. We will deal to the trivial module $V = \mathbb{K}$ case only. In particular, we introduced the generalised page zero \eqref{GCB} that keeps both the modules described above into account. As we commented in \eqref{GCFA} and \eqref{GCFB}, the general page zero \eqref{GCB} can be simplified: if $\dim \mathfrak{h}_1 \neq \dim \mathfrak{k}_1$
\begin{eqnarray}
	\label{GCOA} l = \dim \mathfrak{h}_1 \ &\implies& \mathcal{E}_0^{m,n} = \tilde{E}_0^{m,n} = C^{(m|0)} \left( \mathfrak{k} \right) \otimes C^{(n|\dim \mathfrak{h}_1)} \left( \mathfrak{h} \right) \ , \\
	\label{GCOB} l = \dim \mathfrak{k}_1 \ &\implies& \mathcal{E}_0^{m,n} = E_0^{m,n} = C^{(m|\dim \mathfrak{k}_1)} \left( \mathfrak{k} \right) \otimes C^{(n|0)} \left( \mathfrak{h} \right) \ .
\end{eqnarray}
If $\dim \mathfrak{h}_1 = \dim \mathfrak{k}_1 = \dim \mathfrak{g}_1/2$, we have
\begin{equation}\label{GCOC}
	\mathcal{E}_0^{m,n} = \left[ C^{(m|\dim \mathfrak{g}_1/2)} \left( \mathfrak{k} \right) \otimes C^{(n|0)} \left( \mathfrak{h} \right) \right] \oplus \left[ C^{(m|0)} \left( \mathfrak{k} \right) \otimes C^{(n|\dim \mathfrak{g}_1/2)} \left( \mathfrak{h} \right) \right] \ .
\end{equation}

The computation of the following pages follows in exact analogy to the example presented in the previous section. Hence, the following Proposition:
\begin{prop}\label{KSHtheoremextended}
	Let $\mathfrak{g}$ be a Lie (super)algebra over a (characteristic-zero) field $\mathbb{K}$ and let $\mathfrak{h}$ be a Lie sub-(super)algebra reductive in $\mathfrak{g}$, $\dim \mathfrak{h}_1 \neq 0$, and denote $\mathfrak{k}=\mathfrak{g}/\mathfrak{h}$. Then, at picture number $l=\dim \mathfrak{h}_1,\dim \mathfrak{k}_1$, if $\dim \mathfrak{h}_1 \neq \dim \mathfrak{k}_1$, the first pages of the extended spectral sequence read
	\begin{eqnarray*}
		\label{GCOD} l = \dim \mathfrak{h}_1 \ &\implies& \ \mathcal{E}_1^{m,n} = \left( C^{(m|0)} \left( \mathfrak{k} \right) \right)^{\mathfrak{h}} \otimes H^{(n|\dim \mathfrak{h}_1)} \left( \mathfrak{h} \right) \ , \ \mathcal{E}_2^{m,n} = H^{(m|0)} \left( \mathfrak{g} , \mathfrak{h} \right) \otimes H^{(n|\dim \mathfrak{h}_1)} \left( \mathfrak{h} \right) \ , \\
		\label{GCOE} l = \dim \mathfrak{k}_1 \ &\implies& \ \mathcal{E}_1^{m,n} = \left( C^{(m|\dim \mathfrak{k}_1)} \left( \mathfrak{k} \right) \right)^{\mathfrak{h}} \otimes H^{(n|0)} \left( \mathfrak{h} \right) \ , \ \mathcal{E}_2^{m,n} = H^{(m|\dim \mathfrak{k}_1)} \left( \mathfrak{g} , \mathfrak{h} \right) \otimes H^{(n|0)} \left( \mathfrak{h} \right) \ .
	\end{eqnarray*}
	If $\dim \mathfrak{h}_1 = \dim \mathfrak{k}_1 = \dim \mathfrak{g}_1/2$, then the first two pages at picture number $l = \dim \mathfrak{g}_1/2$ read
	\begin{eqnarray*}
		\label{GCOF} \mathcal{E}_1^{m,n} &=& \left[ \left( C^{(m|0)} \left( \mathfrak{k} \right) \right)^{\mathfrak{h}} \otimes H^{(n|\dim \mathfrak{g}_1/2)} \left( \mathfrak{h} \right) \right] \oplus \left[ \left( C^{(m|\dim \mathfrak{g}_1/2)} \left( \mathfrak{k} \right) \right)^{\mathfrak{h}} \otimes H^{(n|0)} \left( \mathfrak{h} \right) \right] \ , \\
		\label{GCOG} \mathcal{E}_2^{m,n} &=& \left[ H^{(m|0)} \left( \mathfrak{g} , \mathfrak{h} \right) \otimes H^{(n|\dim \mathfrak{g}_1/2)} \left( \mathfrak{h} \right) \right] \oplus \left[ H^{(m|\dim \mathfrak{g}_1/2)} \left( \mathfrak{g} , \mathfrak{h} \right) \otimes H^{(n|0)} \left( \mathfrak{h} \right) \right] \ .
	\end{eqnarray*}
\end{prop}

\begin{proof}
	The proof actually does not differ significantly to the argument we used for the specific example of Sec.\ref{Cohomology via Spectral Sequences}. For the sake of clarity, here we assume $\dim \mathfrak{h}_1 = \dim \mathfrak{k}_1 = \dim \mathfrak{g}_1/2$, the proof in the complementary case follows analogously. The general CE differential reads as in \eqref{TSST}, but the reductivity hypothesis excludes the last term. The first differential for the construction of the spectral sequence, the horizontal one, formally reads as
	\begin{equation}\label{GCP}
		d_0 = \mathcal{Y}^*_{\mathfrak{h}} \mathcal{Y}^*_{\mathfrak{h}} \iota_{\mathcal{Y}_{\mathfrak{h}}} + \mathcal{Y}^*_{\mathfrak{k}} \mathcal{Y}^*_{\mathfrak{h}} \iota_{\mathcal{Y}_{\mathfrak{k}}} \ .
	\end{equation}
	Page 1 of the spectral sequence is naturally defined as in Def.\ref{Generalised KHS Spectral Sequence}:
	\begin{equation}\label{GCQ}
		\mathcal{E}_1^{m,n} \defeq H \left( \mathcal{E}_0^{m,n} , d_0 \right) \ .
	\end{equation}
	We directly see that the part $\mathcal{Y}^*_{\mathfrak{h}} \mathcal{Y}^*_{\mathfrak{h}} \iota_{\mathcal{Y}_{\mathfrak{h}}}$ of $d_0$ selects the cohomology spaces $\displaystyle H^{(n|\bullet)} \left( \mathfrak{h} \right) , \bullet = 0 , \dim \mathfrak{g}_1/2 $. On the other hand, the part $\mathcal{Y}^*_{\mathfrak{k}} \mathcal{Y}^*_{\mathfrak{h}} \iota_{\mathcal{Y}_{\mathfrak{k}}}$ selects the invariant forms among the forms in $\mathfrak{k}$, i.e., $\displaystyle \left( C^{(m|\bullet)} \left( \mathfrak{k} \right) \right)^{\mathfrak{h}} , \bullet = 0 , \dim \mathfrak{g}_1/2$. Then we have
	\begin{equation}\label{GCR}
		\mathcal{E}_1^{m,n} = \left[ \left( C^{(m| \dim \mathfrak{g}_1/2)} \left( \mathfrak{k} \right) \right)^{\mathfrak{h}} \otimes H^{(n|0)} \left( \mathfrak{h} \right) \right] \oplus \left[ \left( C^{(m|0)} \left( \mathfrak{k} \right) \right)^{\mathfrak{h}} \otimes H^{(n| \dim \mathfrak{g}_1/2)} \left( \mathfrak{h} \right) \right] \ .
	\end{equation}
	Notice that both spaces of \eqref{GCR} are zero for $n < 0$, if $\mathfrak{h}$ is a classical basic Lie (super)algebra: $H^{(n|0)} \left( \mathfrak{h} \right) = \left\lbrace 0 \right\rbrace , \forall n < 0$ follows trivially from the fact that superforms are bounded from below; $H^{(n|\dim \mathfrak{g}_1/2)} \left( \mathfrak{h} \right) = \left\lbrace 0 \right\rbrace , \forall n < 0$ follows from the fact that $H^{(n|0)} \left( \mathfrak{h} \right) \neq \left\lbrace 0 \right\rbrace \iff n \leq \text{dim}_0 \left( \mathfrak{h} \right)$ for any classical basic Lie superalgebra $\mathfrak{g}$, then, from the Berezinian complement isomorphism in Prop.\ref{Berezinian Complement Quasi-Isomorphism} it follows that $H^{(n|\dim \mathfrak{g}_1/2)} \left( \mathfrak{h} \right) \neq \left\lbrace 0 \right\rbrace \iff n \geq 0$. This means that in the tabular representation of page 1, the second and third quadrants are empty.

	The following step is the definition of page 2: the differential $d_1$, i.e., the vertical differential, for reductive sub-(super)algebras formally reads as
	\begin{equation}\label{GCS}
		d_1 = \mathcal{Y}^*_{\mathfrak{k}} \mathcal{Y}^*_{\mathfrak{k}} \iota_{\mathcal{Y}_{\mathfrak{k}}} \ ,
	\end{equation}
	and page 2 is defined as
	\begin{equation}\label{GCT}
		\mathcal{E}_2^{m,n} \defeq H \left( \mathcal{E}_1^{m,n} , d_1 \right) \ .
	\end{equation}
	The computation of page 2 is an hard task in general (as emphasised in \cite{Koszul,HocSer}), but it simplifies greatly in our case, since the second differential simply corresponds to the relative differential for the coset $\mathfrak{k}$. For reductive Lie sub-superalgebras we have that $d_1$ acts on $\displaystyle \left( C^{(m|0)} \left( \mathfrak{k} \right) \right)^{\mathfrak{h}} $ and $\displaystyle \left( C^{(m|\dim \mathfrak{g}_1/2)} \left( \mathfrak{k} \right) \right)^{\mathfrak{h}}$ only, giving the relative cohomologies. Hence, page 2 reads as
	\begin{equation}\label{GCU}
		\mathcal{E}_2^{m,n} = \left[ H^{(m|\dim \mathfrak{g}_1/2)} \left( \mathfrak{g} , \mathfrak{h} \right) \otimes H^{(n|0)} \left( \mathfrak{h} \right) \right] \oplus \left[ H^{(m|0)} \left( \mathfrak{g} , \mathfrak{h} \right) \otimes H^{(n|s)} \left( \mathfrak{h} \right) \right] \ .
	\end{equation}
	The cases when $\dim \mathfrak{h}_1 \neq \dim \mathfrak{k}_1$ follow with analogous manipulations.
\end{proof}

We want to emphasise that the description of page zero and the results obtained in Prop.\ref{KSHtheoremextended} mirror directly the results reported in theorems 15.1, 15.2 and 15.3 of \cite{Koszul}, respectively. Really, for reductive Lie sub-superalgebras, we have that \eqref{GCU} is the super Lie algebra extension of page 2 as calculated for Lie algebras (w.r.t. reductive Lie sub-algebras), constructed with the relative cohomology spaces and the cohomology spaces of the sub-algebra.

\section{Cohomology via Distributional Realisation:$\mathfrak{osp}(2|2)$}

After the formal discussion on the cohomology, we present here complete and explicit computations of it for the example of $\mathfrak{osp}(2|2)$ in terms of the distributional realisation described in Sec.\ref{Preliminaries}. We first discuss the result using the associated Poincar\'e polynomials and then we provide an explicit construction of the class representatives.

\subsection{Poincar\'e Polynomials}

We briefly review the definition of Poincar\'e series and Poincar\'e polynomials. This is not strictly necessary for the example under examination, but it paves the way for a future article \cite{In preparation}, where we will discuss how to extend Molien-Weyl integral formula (see, e.g., \cite{FultonHarris}) to the super-setting, in order to calculate the Poincar\'e series \emph{a priori}, and then use it as a guide to infer cohomology classes, for any complex of forms.

\begin{defn}
	Given $X$ a \emph{graded} $\mathbb{K}$-vector space over a (characteristic-zero) field $\mathbb{K}$ with direct decomposition into $p$-degree homogeneous subspaces given by $X = \bigoplus_{p \in \mathbb{Z}} X_p$, we call the formal series 
	\begin{eqnarray}
	\label{POIA}
		\mathpzc{P}_X(t) = \sum_{p\in \mathbb{Z}} ({\rm dim}_k \, X_p) (-t)^p
	\end{eqnarray}
	the \emph{Poincar\'e series} of $X$. Notice that we have implicitly assumed that $X$ is of \emph{finite type}, \emph{i.e.} its homogeneous subspaces $X_p$ are finite dimensional for every $p.$ The unconventional sign in $(-t)^p$ takes into account the \emph{parity} of $X_p$, which takes values in $\mathbb{Z}_2$ and it is given by $p \, \mbox{mod}\, 2$; this is particularly useful for super algebras.
\end{defn}

	If we assume that the pair $(M, \delta)$ is a differential (graded) complex for $M=\oplus_{p\in\mathbb{Z}} M_p$ and $\delta : M_p \rightarrow M_{p+1}$ for any $p$, then $H^\bullet (M , \delta) = \bigoplus_{p\in \mathbb{Z}} H^p(M , \delta)$ is a graded space. We call the numbers $b_p (M) \defeq \dim_\mathbb{K} H^p (M , \delta)$ \emph{Betti numbers}, in alaogy to the usual Betti numbers of a given manifold and its de Rham cohomology. In our algebraic context, given a Lie (super)algebra, we denote with these numbers the dimension of its Chevalley-Eilenberg $p$-cohomology groups $b_p (\mathfrak{g}) = \dim_\mathbb{K} H^p (\mathfrak{g}),$ so that the Poincar\'e series of the Lie (super)algebra $\mathfrak{g}$ is the generating function of its Betti numbers:
	\bear
		\mathpzc{P}_{\mathfrak{g}} (t) = \sum_p b_p (\mathfrak{g}) (-t)^p.
	\eear
	Notice that we used the word \virgolette series" on purpose: indeed, $H^\bullet (\mathfrak{g})$ is not in general finite-dimensional for a generic Lie superalgebra $\mathfrak{g}$ (see, e.g., \cite{CCGN2}). In Sec.\ref{Pseudoforms as Infinite-Dimensional Representations}, we introduced the notion of \virgolette picture number", as an analogous way to specify with which kind of forms (superforms, pseudoforms, integral forms) one is dealing, or, in other words, which $\mathfrak{g}$-module one is considering. To include the picture number in the Poincar\'e series, it is useful to introduce a second grading. In this case, the cohomology groups are indicated as $H^{(p|q)} \left( \mathfrak{g} \right)$, where $p$ is the form number and $q$ is the picture number, so that the Poincar\'e series reads
\begin{eqnarray}
\label{POIE}
\mathpzc{P}_{\mathfrak{g}}(t,\tilde t) = \sum_{p,q} (-t)^p (-\tilde t)^q {\rm dim} H^{(p|q)} \left( \mathfrak{g} \right) \ .
\end{eqnarray}
This allows an easy identification of cohomology spaces, as $t$ counts the form number and $\tilde{t}$ counts the picture number.

First, let us recollect some results from \cite{CE}. If one considers the bosonic sub-algebra $\mathfrak{so}(2) \oplus \mathfrak{sp}(2)$ of 
$\mathfrak{osp}(2|2)$, its Poincar\'e polynomial is factorized into a product of two polynomials (as follows from the K\"unneth formula) 
\begin{eqnarray}
\label{PPA}
\mathpzc{P}_{\mathfrak{so}(2) \oplus \mathfrak{sp}(2)}(t) = (1- t) (1-t^3) \ ,
\end{eqnarray}
counting both the cohomology classes of the abelian factor $\mathfrak{so}(2)$ and those of the non-abelian one $\mathfrak{sp}(2)$. In the same 
way, for the (super)algebra $\mathfrak{osp}(1|2)$ we have 
\begin{eqnarray}
\label{PPB}
\mathpzc{P}_{\mathfrak{osp}(1|2)}(t,\tilde t) = (1-t^3) (1 + \tilde t^2)   = (1-t^3) +  (1-t^3) \tilde t^2\,,  
\end{eqnarray}
where the polynomial $(1-t^3)$ counts the cohomology spaces of superforms $H^{(\bullet|0)}(\mathfrak{osp}(1|2))$ while 
$(1-t^3) \tilde t^2$ takes into account the cohomology spaces of integral forms $H^{(\bullet|2)} (\mathfrak{osp}(1|2))$. 

Poincar\'e series turn out to be a particularly useful tool when dealing with cosets: in \cite{GHV} there is a theorem that allows to calculate the Poincar\'e series of certain cosets: given a Lie algebra $\mathfrak{g}$ with Poincar\'e polynomial $\mathpzc{P}_\mathfrak{g} (t) = \sum_i \left( 1 - t^{c^{\mathfrak{g}}_i} \right)$, where $c^{\mathfrak{g}}_i$ are the usual exponents in the factorised form of the polynomial, and $\mathfrak{h}$ a Lie sub-algebra of $\mathfrak{g}$, of the same rank, with Poincar\'e polynomial given by $\mathpzc{P}_\mathfrak{h} (t) = \sum_i \left( 1 - t^{c^{\mathfrak{h}}_i} \right)$, then the Poincar\'e polynomial for the coset $\mathfrak{g}/\mathfrak{h}$ will be given by
\bear \label{cosetp}
\mathpzc{P}_{\mathfrak{g}/\mathfrak{h}}(t) = \frac{\prod_{i} (1-t^{c^\mathfrak{g}_i+1})}{\prod_{j} (1-t^{c^\mathfrak{h}_j+1})} \ .
\eear
Actually, since we are dealing with superalgebras, we have to adapt the previous notions in order to keep into account for the picture number. This means that the Poincar\'e polynomial of $\mathfrak{g}$ and $\mathfrak{h}$ will read in general $\mathpzc{P}_\mathfrak{g} (t, \tilde{t}) = \sum_i \left( 1 - t^{c^{\mathfrak{g}}_i} \tilde{t}^{p^{\mathfrak{g}}_i} \right)$ and $\mathpzc{P}_\mathfrak{h} (t, \tilde{t}) = \sum_i \left( 1 - t^{c^{\mathfrak{h}}_i} \tilde{t}^{p^{\mathfrak{h}}_i} \right)$, respectively. \eqref{cosetp} is extended to superalgebras as
\begin{equation}\label{cosetpA}
	\mathpzc{P}_{\mathfrak{g}/\mathfrak{h}}(t, \tilde{t}) = \frac{\prod_{i} \left( 1-t^{c^\mathfrak{g}_i+1} \tilde{t}^{p^{\mathfrak{g}}_i} \right)}{\prod_{j} \left(1-t^{c^\mathfrak{h}_j+1} \tilde{t}^{p^{\mathfrak{g}}_j} \right)} \ .
\end{equation}
This product formula is very helpful since it provides information regarding the dimension of the cohomology groups of the coset, as we will show for our main example in the following. In \cite{In preparation} we will also show how to use this formula in the opposite direction: from the cohomology of the coset $\mathfrak{g}/\mathfrak{h}$ and of the sub-(super)algebra $\mathfrak{h}$, we can reconstruct the Poincar\'e series of the algebra $\mathfrak{g}$ and hence have directly the information about the dimension of its cohomology groups.

\subsection{Explicit Construction}

Let us now move to the main example. Collecting the results of the Sec.\ref{Cohomology via Spectral Sequences} displayed in Prop.\ref{osp2|2} into the Poincar\'e polynomial of $\mathfrak{osp}(2|2)$, we get 
\begin{eqnarray}
\label{PPD}
\mathpzc{P}_{\mathfrak{osp}(2|2)}(t,\tilde t) &=& (1- t^3) (1 - t \tilde t^2) (1 + \tilde t^2) = \nonumber \\&=& (1- t^3) - (1 - t^3) t \tilde t^2 + (1-t^3) \tilde t^2  - (1-t^3) t \tilde t^4 \ .
\end{eqnarray}
Interpreting the polynomial in the second line, the first parenthesis represents superform cohomology $H^{(p|0)} \left( \mathfrak{osp} \left( 2|2 \right) \right)$ with $p=0,3$; the second term $- t \tilde t^2 + t^4 \tilde t^2$ corresponds to the pseudoform cohomology $H^{(p|2)}$, with $p=1,4$, discussed in \eqref{TSSY}
and selected with the filtration \eqref{TSSA}. The third term $\tilde t^2 -t^3 \tilde t^2$ corresponds to the pseudoform cohomology selected by the inequivalent 
filtration \eqref{TSSZA} and discussed in \eqref{TSSZM}. Finally, the last polynomial $-  t \tilde t^4 + t^4\tilde t^4$ counts the cohomology of 
integral forms $H^{(p|4)} \left( \mathfrak{g} \right)$. In addition, we have noticed that 
the Berezinian complement duality holds among all form complexes: $\star H^{(p|q)} \left( \mathfrak{g} \right) = H^{(4-p|4-q)} \left( \mathfrak{g} \right)$. This translates in the polynomial as
\begin{equation}
	\mathpzc{P}_{\mathfrak{osp}(2|2)}(t,\tilde t) = \mathpzc{P}_{\mathfrak{osp}(2|2)}(1/t,1/\tilde{t}) t^4 \tilde{t}^4 \ .
\end{equation}

Let us fix the sub-algebra $\mathfrak{h} = \mathfrak{g}_0 = \mathfrak{so}(2) \times \mathfrak{sp}(2)$. We now study the coset space $\mathfrak{k} = \mathfrak{g/h}$  by using the explicit distributional realisation of pseudoforms and integral forms, in order to describe the cohomology of $\mathfrak{g}$. For that, we compute the Poincar\'e polynomial of $\mathfrak{k}$ using \eqref{cosetpA}:
\begin{eqnarray}
\label{PPE}
\mathpzc{P}_{\mathfrak{k}}(t,\tilde t) 
&=& \frac{(1- t^4) (1 - t^2 \tilde t^2) (1 + \tilde t^2)}{(1-t^2) (1-t^4)} 
=  \frac{(1 - t^2 \tilde t^2) (1 + \tilde t^2)}{(1-t^2)} =
\nonumber \\
&=& 
 \frac{1 + (1-t^2) t^2 \tilde t^2  - t^2 \tilde t^4}{(1-t^2)} = 
\frac{1}{1-t^2} + \tilde t^2 + \frac{1}{1- \frac{1}{t^2}} \tilde t^4 \ .
\end{eqnarray}
The coset space $\mathfrak{k}$ is a $(0|4)$-dimensional space, with no bosonic generators. Therefore the MC forms $\psi^\pm, \bar\psi^\pm$ are covariantly constant, as the covariant differential $\nabla$ is trivial, so that the relative cohomology $H(\mathfrak{g} , \mathfrak{h})$ is easily computed as it is represented by invariant forms in $\mathfrak{k}$ only. The three pieces of the series \eqref{PPE} corresponds to $H^{(\bullet|0)} (\mathfrak{g} ,  \mathfrak{h}), H^{(\bullet|2)} (\mathfrak{g} ,  \mathfrak{h})$ and $H^{(\bullet|4)} (\mathfrak{g} ,  \mathfrak{h})$ and here we discuss them explicitly. 
 
The cohomology of superforms $H_{super}^\bullet(\mathfrak{k})$ is generated  by any power of the $(2|0)$-form (which is basic, i.e., horizontal and $\mathfrak{h}$-invariant)
\begin{eqnarray}
\label{PPF}
K^{(2|0)} = \psi^+\wedge \bar{\psi}^- - \bar{\psi}^+ \wedge \psi^- \ .
\end{eqnarray}
Note that, according to \eqref{VSPJ}, $K^{(2|0)} \propto dU$, so that $\mathfrak{h}$-invariance is straightforwardly verified. Since $K^{(2|0)}$ is an even cohomology representative, we should consider any power of it as well\footnote{The reductivity of $\mathfrak{h}$ in $\mathfrak{g}$ induces a ring structure on the relative cohomology. See, e.g., \cite{Koszul}.}. They are taken into account in the Poincar\'e series as $\sum_{p\leq 0} t^{2p} = 1/(1- t^2)$, corresponding to the first term of \eqref{PPE}.

$K_2$ can also be written as $K_2 = \frac12 \epsilon^{\a\b} \epsilon_{IJ} \psi^I_\alpha \psi^J_\beta$, where $\psi^I_\alpha$ are the MC forms in the real representation ($I,J=1,2$ and 
$\alpha, \beta =1,2$):
\begin{equation}
	\psi^1_1 = \psi^+ \ , \ \psi^1_2 = \psi^- \ , \ \psi^2_1 = \bar{\psi}^+ \ , \ \psi^2_2 = \bar{\psi}^- \ .
\end{equation} 

The cohomology of integral forms is computed by using the Berezinian duality prescription as follows: we start from the 
top integral form of the supercoset $\mathfrak{k}$, which explicitly reads
\begin{eqnarray}
\label{PPG}
\omega^{(0|4)} = \delta(\psi^+) \delta(\psi^-) \delta(\bar\psi^+) \delta(\bar\psi^-) \ .
\end{eqnarray}
\eqref{PPG} is trivially verified to be basic. Again, since the covariant derivative is trivial, we should select the basic integral forms only. By respecting the symmetry of the sub-algebra, we define the double-contraction 
operator $\iota_2$ as
\begin{eqnarray}
\label{PPH}
\iota_2 = \iota_{F^+} \iota_{\bar{F}^-} - \iota_{\bar F^+} \iota_{F^-} \ ,
\end{eqnarray}
where $F^\pm$ and $\bar F^\pm$ are the odd generators of the 
superalgebra. $\iota_2$ is defined to be the formal inverse (modulo multiplication by constants) of $K^{(2|0)}$. We can act with any power $\iota_2^p$ on $\omega^{(0|4)}$ 
to get the infinite number of cohomology representatives 
\begin{eqnarray}
\label{PPI}
\iota^{p}_2 \delta(\psi^+) \delta(\psi^-) \delta(\bar\psi^+) \delta(\bar\psi^-) \ , \forall p \geq 0 \ ,
\end{eqnarray}
which generate the complete integral form cohomology
$H^{(\bullet|4)} \left( \mathfrak{g} , \mathfrak{h} \right)$. Indeed, since $\iota_2^p$ corresponds to 
$t^{-2p}$ (since $\iota_2$ has \virgolette form number" -2) , the Poincar\'e series for integral forms reads 
\begin{eqnarray}
\label{PPJ}
\sum_{p=0}^\infty t^{-2p} \tilde{t}^4 = \frac{1}{(1 - \frac{1}{t^2})} \tilde t^4 \ . 
\end{eqnarray}

Let us come to the final piece of \eqref{PPE}, namely the single term $\tilde t^2$. 
This monomial suggests two possible scenarios: on one hand, it could suggest that there is a single cohomology generator among pseudoforms, with zero form number; on the other hand, it could seem to suggest that there is both such generator and a whole tower of pseudoforms obtained by multiplying (thanks to the ring structure) this generator by $\left( K^{(2|0)} \right)^p , \forall p \geq 0$ or by acting on it with the operator $\iota_2^p , \forall p \geq 0$. In this case, the corresponding term of the Poincar\'e series would read $\displaystyle \frac{\tilde{t}^2}{1-t^2} + \frac{\tilde{t}^2}{1-\frac{1}{t^2}} = \tilde{t}^2 $. As we will explicitly see, the multiplication by $K^{(2|0)}$ or the action of $\iota_2$ on such pseudoform gives zero, so that the whole pseudoform cohomology is generated by a single term.
 
To single out the pseudoform with null form number, we have to construct 
an object which is $\mathfrak{so}(2) \oplus \mathfrak{sp}(2)$ invariant. 
By using the real representation $\psi^I_\alpha$, we observe that the combinations 
\begin{eqnarray}
\label{PPM}
\eta(\psi^1) = \epsilon_{\a\b}\delta(\psi^1_\alpha) \delta(\psi^1_\beta)\,, ~~~~
\eta(\psi^2) = \epsilon_{\a\b}\delta(\psi^2_\alpha) \delta(\psi^2_\beta)\,, ~~~~
\end{eqnarray}
are invariant w.r.t. the sub-algebra $\mathfrak{sp}(2)$: given $X \in \mathfrak{sp}(2)$, we have
\begin{equation}
	\mathcal{L}_X \eta(\psi^1) = \left( \iota_X d + d \iota_X \right) 2 \delta(\psi^1_1) \delta(\psi^1_2) = \iota_X \left[ U \psi^2_1 \iota^1_1 \delta(\psi^1_1) \delta(\psi^1_2) + U \psi^2_2 \iota^1_2 \delta(\psi^1_1) \delta(\psi^1_2) \right] = 0 \ ,
\end{equation}
where we have used \eqref{VSPG} $\div$ \eqref{VSPK}. The same holds true for $\eta(\psi^2)$.

The invariance under $\mathfrak{so}(2)$ is implemented by 
requiring
\begin{eqnarray}
\label{PPN}
\sum_{\alpha=1,2} \Big( 
\psi^1_\alpha \iota_{F^2_\alpha} - \psi^2_\alpha \iota_{F^1_\alpha} 
\Big) \sigma(\psi^1, \psi^2) =0 \ ,
\end{eqnarray}
where $F^I_\alpha$ are the dual vectors to $\psi^I_\alpha$: $ \iota_{F^I_\alpha} \psi^J_\beta = \delta^J_I \delta^\alpha_\beta$. The generalized form $\sigma(\psi^1, \psi^2)$ satisfying \eqref{PPN}
depends on $\psi^I_\alpha$ through the combinations in \eqref{PPM}. The solution of \eqref{PPN}, seen as a differential equation, can be expressed in terms of the formal
$0^{th}$-order Bessel function $J_0(x) = 1 - x^2/4 + x^4/64 + \dots $ as follows:
\begin{eqnarray}
\label{PPP}
\sigma(\psi^1, \psi^2) = J_0
\Big( \sum_{\alpha} \psi^2_\a \iota_{F^1_\alpha} \Big) \eta(\psi^1) = 
\eta(\psi^1) - \frac14  \sum_{\alpha,\beta} \psi^2_\a  \psi^2_\beta \iota_{F^1_\alpha} 
\iota_{F^1_\beta}   \eta(\psi^1) + \dots \ .
\end{eqnarray}
The invariant pseudoform $\sigma(\psi^1, \psi^2)$ can also be expanded around 
$\eta(\psi^2)$, with an analogous expression that exchanges $\psi^1$ with $\psi^2$:
\begin{equation}\label{PPPA}
	\sigma(\psi^1, \psi^2) = J_0
\Big( \sum_{\alpha} \psi^1_\a \iota_{F^2_\alpha} \Big) \eta(\psi^2) = 
\eta(\psi^2) - \frac14  \sum_{\alpha,\beta} \psi^1_\a  \psi^1_\beta \iota_{F^2_\alpha} 
\iota_{F^2_\beta} \eta(\psi^2) + \dots \ .
\end{equation} 
The equivalence of the two representations is consistent with the duality for the coset, since the zero form pseudoform counted by the monomial $\tilde t^2$ in the Poincar\'e polynomial \eqref{PPD} is self-dual, as explicitly shown with the representatives \eqref{PPP} and \eqref{PPPA}. This is not the end of the story: we should verify if it is possible to construct a whole tower of pseudoforms in cohomology by means of the generators \eqref{PPF}, \eqref{PPH} and \eqref{PPP}. Namely, we can multiply $K^{(2|0)}$ by $\sigma(\psi^1, \psi^2)$ and then we have to verify if this reproduces a new cohomology class:
\begin{equation}\label{PPPB}
	K^{(2|0)} \wedge \sigma(\psi^1, \psi^2) = \frac12 \epsilon^{\a\b} \epsilon_{IJ} \psi^I_\alpha \psi^J_\beta\wedge \eta(\psi^1) - \frac18 \epsilon^{\gamma \delta} \epsilon_{IJ} \psi^I_\gamma \psi^J_\delta \wedge \sum_{\alpha,\beta} \psi^2_\a \psi^2_\beta \iota_{F^1_\alpha} \iota_{F^1_\beta} \eta(\psi^1) + \ldots \ .
\end{equation}
The first term in \eqref{PPPB} is trivially zero, all the higher terms vanish independently, thanks to a cancellation due to the minus sign contained in $K^{(2|0)}$. For example, the second term reads
\begin{equation}\label{PPPC}
	- \frac18 \epsilon^{\gamma \delta} \epsilon_{IJ} \psi^I_\gamma \psi^J_\delta \wedge \sum_{\alpha,\beta} \psi^2_\a \psi^2_\beta \iota_{F^1_\alpha} \iota_{F^1_\beta} \eta(\psi^1) = $$ $$ = - \frac14 \left( \psi^1_1 \psi^2_2 - \psi^2_1 \psi^1_2 \right) \wedge \left[ \left( \psi^2_1 \right)^2 \iota^2_{F^1_1} + 2 \psi^2_1 \psi^2_2 \iota_{F^1_1} \iota_{F^1_2} + \left( \psi^2_2 \right)^2 \iota^2_{F^1_2} \right] \eta(\psi^1) = $$ $$ = \frac{1}{2} \left[ \left( \psi^2_1 \right)^2 \psi^2_2 \iota_{F^1_1} + \psi^2_1 \left( \psi^2_2 \right)^2 \iota_{F^1_2} - \left( \psi^2_1 \right)^2 \psi^2_2 \iota_{F^1_1} - \psi^2_1 \left( \psi^2_2 \right)^2 \iota_{F^1_2} \right] \eta(\psi^1) = 0 \ ,
\end{equation}
having used the properties of the deltas, listed in \eqref{IFLSF}. It is easy to see, with analogous manipulations, that all the higher terms of \eqref{PPPB} vanish as well. On the other hand, one could ask if it is possible to generate new cohomology classes by acting with the operator $\iota_2 = \iota_{F^1_1} \iota_{F^2_2} - \iota_{F^2_1} \iota_{F^1_2}$ on $\sigma(\psi^1, \psi^2)$; a cancellation analogous to the one in \eqref{PPPB} and \eqref{PPPC} occurs. Really, one has
\begin{equation}\label{PPPD}
	\iota_2 \sigma(\psi^1, \psi^2) = \left( \iota_{F^1_1} \iota_{F^2_2} - \iota_{F^2_1} \iota_{F^1_2} \right) \eta(\psi^1) - \frac14 \left( \iota_{F^1_1} \iota_{F^2_2} - \iota_{F^2_1} \iota_{F^1_2} \right) \wedge \sum_{\alpha,\beta} \psi^2_\a \psi^2_\beta \iota_{F^1_\alpha} \iota_{F^1_\beta} \eta(\psi^1) + \ldots = 0 \ .
\end{equation}
The vanishing of the first term is trivial, because of the contractions along the vectors $F^2_\bullet$, the higher terms vanish because of the minus sign in $\iota_2$, as for \eqref{PPPC}. For example, the second term of \eqref{PPPD} reads
\begin{equation}\label{PPPE}
	- \frac14 \left( \iota_{F^1_1} \iota_{F^2_2} - \iota_{F^2_1} \iota_{F^1_2} \right) \left[ \left( \psi^2_1 \right)^2 \iota^2_{F^1_1} + 2 \psi^2_1 \psi^2_2 \iota_{F^1_1} \iota_{F^1_2} + \left( \psi^2_2 \right)^2 \iota^2_{F^1_2} \right] \eta(\psi^1) = $$ $$ = - \frac12 \left[ \psi^2_1 \iota_{F^1_1}^2 \iota_{F^1_2} + \psi^2_2 \iota^2_{F^1_2} \iota_{F^1_1} - \psi^2_1 \iota^2_{F^1_1} \iota_{F^1_2} - \psi^2_2 \iota_{F^1_1} \iota_{F^1_2}^2 \right] \eta(\psi^1) = 0 \ ,
\end{equation}
and analogously for all the higher terms. This confirms that actually there is a single cohomology generator for pseudoforms in the coset $\mathfrak{osp}(2|2) / \left( \mathfrak{so}(2) \oplus \mathfrak{sp}(2) \right)$.

The results of this section are non-trivial: we have shown that the explicit distributional realisation of pseudoforms and integral forms is a powerful tool to calculate cohomology representatives and that the explicit results are consistent with those obtained without referring to any realisation through spectral sequences. Nonetheless, we have been able to construct a pseudoform that is invariant w.r.t. $\mathfrak{so}(2) \oplus \mathfrak{sp}(2)$; the explicit realisation, supported by the abstract counterpart, could serve as a starting point for the introduction of pseudoforms in more general contexts, e.g., on supermanifolds: one could think to introduce pseudoforms as integral forms of sub-supermanifolds respecting some isometries (see also \cite{Manin:1988ds}). The general setup for these definitions is once again suggested by the algebraic setting analysed in this paper: relative (co)homology. This will be the subject of future investigations.

\section*{Acknowledgements}
\noindent This work has been partially supported by Universit\`a del Piemonte Orientale research funds. We thank S. Cacciatori, R. Catenacci and S. Noja for many useful discussions. C.A.C. is grateful to B. Jur\v co for many suggestions and discussions.

\end{document}